%% file: main.tex
\documentclass[10pt]{article}   
\usepackage{array, xcolor}
\usepackage{amsmath,amsfonts,graphicx, amssymb, placeins, bbm}
\usepackage{subfigure}
\usepackage{url}
\usepackage{xargs}
\usepackage{cite}
\usepackage{breqn}
\usepackage{multirow}
\usepackage{algorithm}
\usepackage[noend]{algpseudocode}
\usepackage{etoolbox}
\usepackage{flushend}
\usepackage{contour}
\usepackage{wrapfig}
\usepackage{geometry}
\usepackage{soul,color}
\usepackage{enumitem}
\soulregister\cite7
\soulregister\ref7
\soulregister\pageref7
\DeclareGraphicsExtensions{.pdf,.png,.jpg}

\usepackage{amsthm}

\makeatletter
\def\thm@space@setup{\thm@preskip=0pt
\thm@postskip=3pt}
\makeatother

\newtheorem{theorem}{Theorem}[section]
\newtheorem{lemma}[theorem]{Lemma}
\newtheorem{proposition}[theorem]{Proposition}

\newtheorem{definition}[theorem]{Definition}

\newenvironment{proof-sketch}{{\bf Proof Sketch:}}{\hfill\rule{2mm}{2mm}}
\newenvironment{note}{\noindent\textbf{Note:}}



\newif\iffullresults
\fullresultstrue

\newcommand\Mark[1]{\textsuperscript#1}

\title{Max-min Fair Rate Allocation and Routing in Energy Harvesting Networks: Algorithmic Analysis}   
\iffullresults\author{
Jelena Mara\v{s}evi\'{c}\Mark{1}, Cliff Stein\Mark{2}, Gil Zussman\Mark{1}\\
       \Mark{1}Department of Electrical Engineering,\\ \Mark{2}Department of Industrial Engineering and Operations Research\\
       {Columbia University}\\
       {\{jelena@ee, cliff@ieor, gil@ee\}.columbia.edu}
} 
\else\author{ACM MobiHoc Submission \#1569892459} 
\fi
\date{}
\graphicspath{{imgs/}}

\newcommand{\littlesum}{\mathop{\textstyle\sum}}
\begin{document}

\maketitle

\begin{abstract}
This \iffullresults paper \fi considers max-min fair rate allocation and routing in energy harvesting networks where fairness is required among \emph{both the nodes and the time slots}. 
Unlike most previous work on fairness, we focus on \emph{multihop topologies} and consider \emph{different routing methods}.
We assume a predictable energy profile and focus on the design of efficient and optimal algorithms that can serve as benchmarks for distributed and approximate algorithms.
We first develop an algorithm that obtains a max-min fair rate assignment for any given (time-variable or time-invariable) \emph{unsplittable routing} or a \emph{routing tree}. For \emph{time-invariable unsplittable routing}, we also develop an algorithm that finds routes that maximize the minimum rate assigned to any node in any slot. For \emph{fractional routing}, we study the joint routing and rate assignment problem. We develop an algorithm for the \emph{time-invariable} case with constant rates. We show that the \emph{time-variable}  case is at least as hard as the 2-commodity feasible flow problem and design an FPTAS to combat the high running time. Finally, we show that finding an unsplittable routing or a routing tree that provides lexigographically maximum rate assignment (i.e., that is the best in the max-min fairness terms) is NP-hard, even for a time horizon of a single slot.  
Our analysis provides insights into the problem structure and can be applied to other related fairness problems.

\noindent {\bf Keywords:} Energy Harvesting, Energy Adaptive Networking, Sensor networks, Routing, Fairness
\end{abstract}

\section{Introduction}
\input{intro.tex}

\section{Model and Problem Formulation}\label{section:system model}
\input{system-model.tex}

\section{Related Work}\label{section:related work}
\input{related-work.tex}

\section{Max-min Fairness and Lexicographic Maximization}\label{section:background}
\input{maxmin-lexmax-background.tex}

\section{Rates in Unsplittable Routing}\label{section:unsplittable}
\input{unsplittable-rates.tex}

\section{Fractional Routing}\label{section:fractional}
\input{fractional.tex}

\section{Fixed Fractional Routing}\label{section:fixed fractional}
\input{fixed-fractional.tex}

\section{Determining a Routing}\label{section:hardness}
\input{hardness.tex}

\section{Conclusions and Future Work}\label{section:conclusion}
\input{conclusion.tex}

\section{Acknowledgements}
{This research was supported in part by NSF grants CCF-1349602, CCF-09-64497, and CNS-10-54856.
The authors are grateful to Prof. Mihalis Yannakakis for useful discussions.}

\vspace{-5pt}
\bibliographystyle{abbrv}
{
\bibliography{references}
}

\end{document}

%% file: intro.tex
\begin{wrapfigure}{r}{0.4\textwidth}
\centering
\vspace{-10pt}\includegraphics[width=0.4\textwidth]{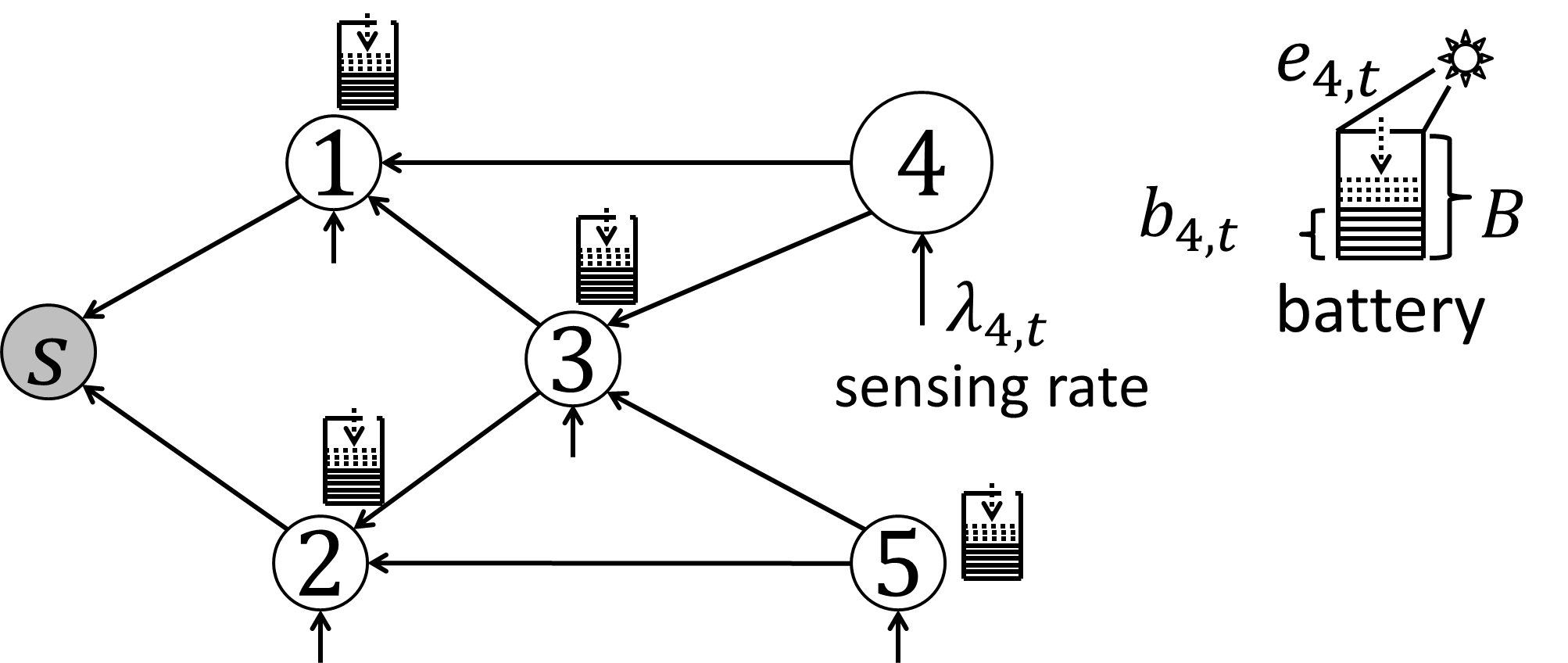} \vspace{-10pt}
\label{fig:example network}
\caption{A simple energy harvesting network: the nodes sense the environment and forward the data to a sink $s$. Each node has a battery of capacity $B$. At time $t$ a node $i$\rq{}s battery level is $b_{i, t}$, it harvests $e_{i, t}$ units of energy, and senses at data rate $\lambda_{i, t}$. }
\vspace{-10pt}
\end{wrapfigure}
Recent advances in the development of ultra-low-power transceivers and energy harvesting devices (e.g., solar cells) will enable self-sustainable and perpetual wireless networks \cite{gorlatova2009challenge, debruin2013monjolo, rob-experimental}. In contrast to legacy wireless sensor networks, where the available energy only decreases as the nodes sense and forward data, in energy harvesting networks the available energy can also increase through a replenishment process. This results in significantly more complex variations of the available energy, which pose challenges in the design of resource allocation and routing algorithms.

The problems of resource allocation, scheduling, and routing in energy harvesting networks have received considerable attention \cite{ozel2011transmission, node-koksal, chen2011finite, maria-wiopt2011, maria-infocom2011, uysal2013finite, dohler-link, gurakan2013energy, shroff-srikant, sarkar-kar, neely, gatzianas, OSU--lexicographic}. Most existing work considers simple networks consisting of a single node or a link {\cite{ozel2011transmission, node-koksal, maria-wiopt2011, maria-infocom2011, uysal2013finite, dohler-link}.} Moreover, fair rate assignment has not been thoroughly studied, and most of the work either focuses on maximizing the total (or average) throughput \cite{ozel2011transmission, node-koksal, chen2011finite, uysal2013finite, dohler-link, shroff-srikant, sarkar-kar, mao-koksal-shroff, neely, gatzianas}, or considers fairness either \emph{only over nodes} \cite{OSU--lexicographic} or \emph{only over time}\cite{maria-wiopt2011, maria-infocom2011}. An exception is \cite{gurakan2013energy}, which requires fairness over both the nodes and the time, but is limited to \emph{two nodes}.

\emph{In this \iffullresults paper\fi, we study the max-min fair rate assignment and routing problems, requiring fairness over both nodes and time slots, and with the goal
of designing optimal and efficient algorithms.}

Following \cite{chen2011finite, maria-wiopt2011, maria-infocom2011, gurakan2013energy, shroff-srikant, OSU--lexicographic}, we assume that the harvested energy is known for each node over a finite time horizon. Such a setting corresponds to a highly-predictable energy profile, and can also be used as a benchmark for evaluating algorithms designed for unpredictable energy profiles.
We consider an energy harvesting  sensor network with a single sink node, and network connectivity modeled by a directed graph (Fig.~\ref{fig:example network}). 
Each node senses some data from its surrounding (e.g., air pressure, temperature, radiation level), and sends it to the sink. The nodes spend their energy on sensing, sending, and receiving data.

\subsection{Fairness Motivation}

Two natural conditions that a network should satisfy are: 
\begin{enumerate}[label=(\roman*), topsep=3pt, itemsep = -2pt]
\item\label{item:fairness-nodes} balanced data acquisition across the entire network, and 
\item\label{item:fairness-time} persistent operation (i.e., even when the environmental energy is not available for harvesting). 
\end{enumerate}

 \begin{wrapfigure}{r}{0.4\textwidth}
\centering
\includegraphics[scale=0.15]{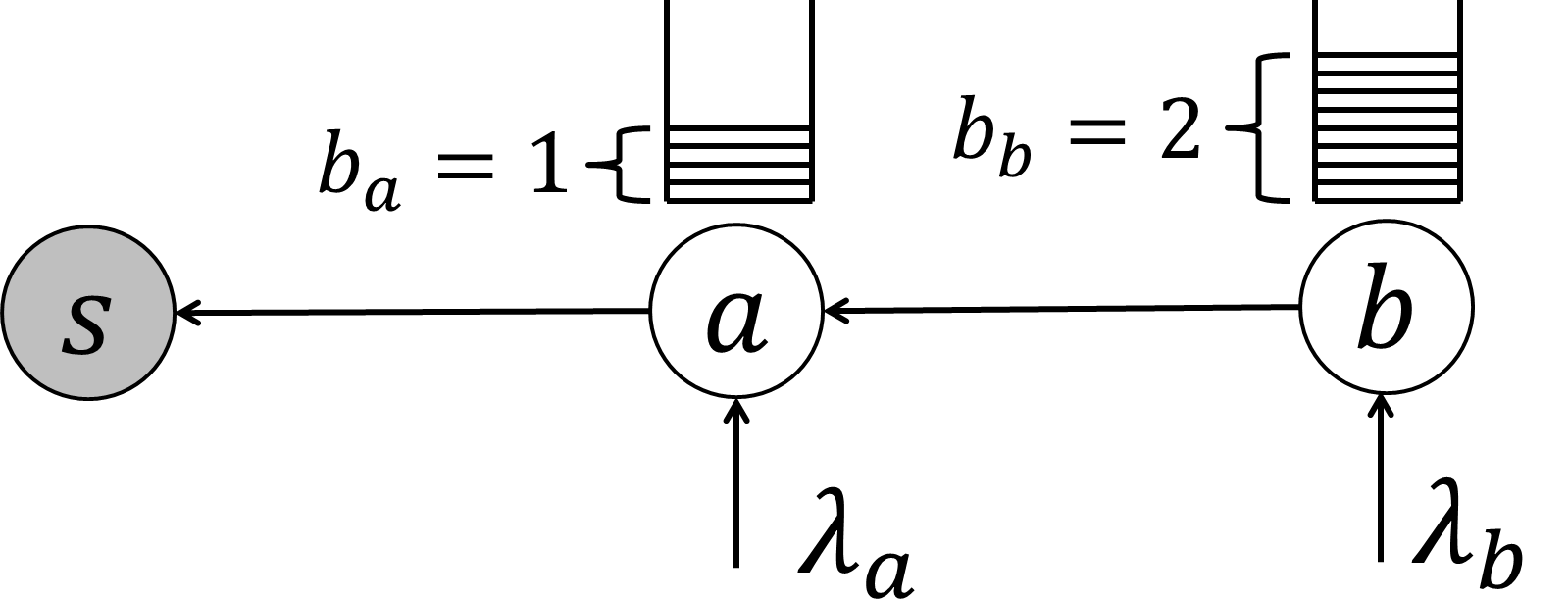}\vspace{-7pt}
\caption{An example of a network in which throughput maximization can result in a very unfair rate allocation among the nodes.}
\label{fig:unfairness}
\vspace{-10pt}
\end{wrapfigure}
The condition \ref{item:fairness-nodes} is commonly reinforced by requiring fairness of the sensing rates over network nodes. We note that in the network model we consider, due to these different energy costs for sending, sensing, and receiving data, throughput maximization can be inherently unfair even in the static case. Consider a simple network with two energy harvesting nodes $a$ and $b$ and a sink $s$ illustrated in Fig.~\ref{fig:unfairness}. Assume that $a$ has one unit of energy available, and $b$ has two units of energy. Let $c_{\text{st}}$ denote the joint cost of sensing and sending a unit flow, $c_{\text{rt}}$ denote the joint cost for receiving and sending a unit flow. Let $\lambda_a$ and $\lambda_{b}$ denote the sensing rates assigned to the nodes $a$ and $b$, respectively. Suppose that the objective is to maximize $\lambda_a+\lambda_b$. If $c_{\text{st}}=1$, $c_{\text{rt}}=2$, then in the optimal solution $\lambda_a=1$ and $\lambda_b=0$. Conversely, if $c_{\text{st}}=2$, $c_{\text{rt}}=1$, then in the optimal solution $\lambda_a=0$ and $\lambda_b=1$. This example easily extends to more general degenerate cases in which throughput-maximum solution assigns non-zero sensing rates only to one part of the network, whereas the remaining nodes do not send any data to the sink.

One approach to achieving \ref{item:fairness-time} is by assigning constant sensing rates to the nodes. However, this approach can result in underutilization of the available energy. As a simple example, consider a node that harvests outdoor light energy over a 24-hour time horizon. If the battery capacity is small, then the sensing rate must be low to prevent battery depletion during the nighttime. However, during the daytime, when the harvesting rates are high, a low sensing rate prevents full utilization of the energy that can be harvested. Therefore, it is advantageous to vary the sensing rates over time. However, fairness must be required over time slots to prevent the rate assignment algorithm from assigning high rates during periods of high energy availability, and zero rates when no energy is available for harvesting. 

\subsection{Routing Types}

We consider different routing types, which are illustrated in Fig.~\ref{fig:model}.

 \begin{figure}[h]
\centering
\hspace{\fill}
\subfigure[Routing tree.]{\label{fig:model--tree}\includegraphics[scale = 0.14]{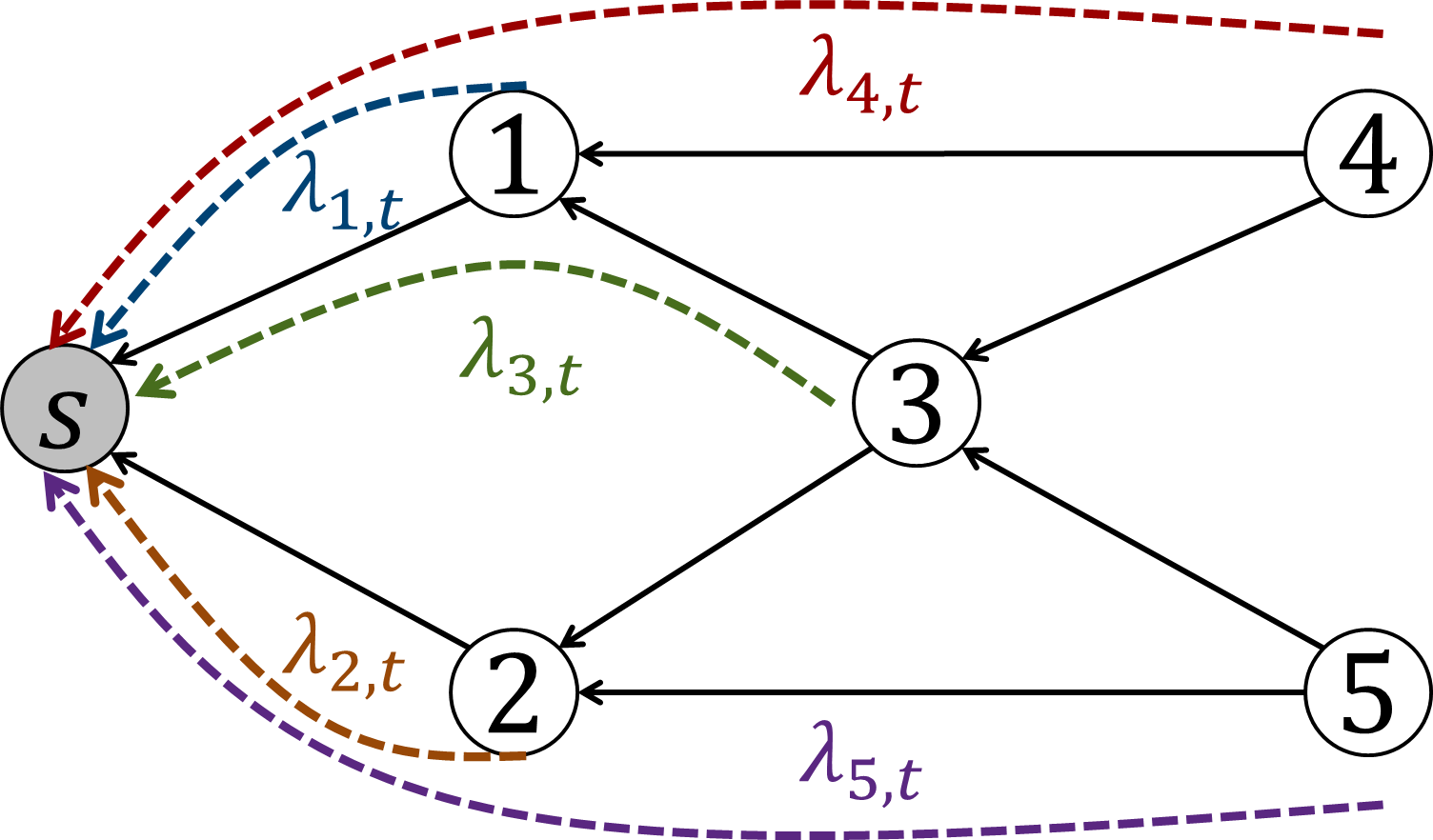}}\hspace{\fill}
\subfigure[Unsplittable routing.]{\label{fig:model--unsplittable}\includegraphics[scale = 0.14]{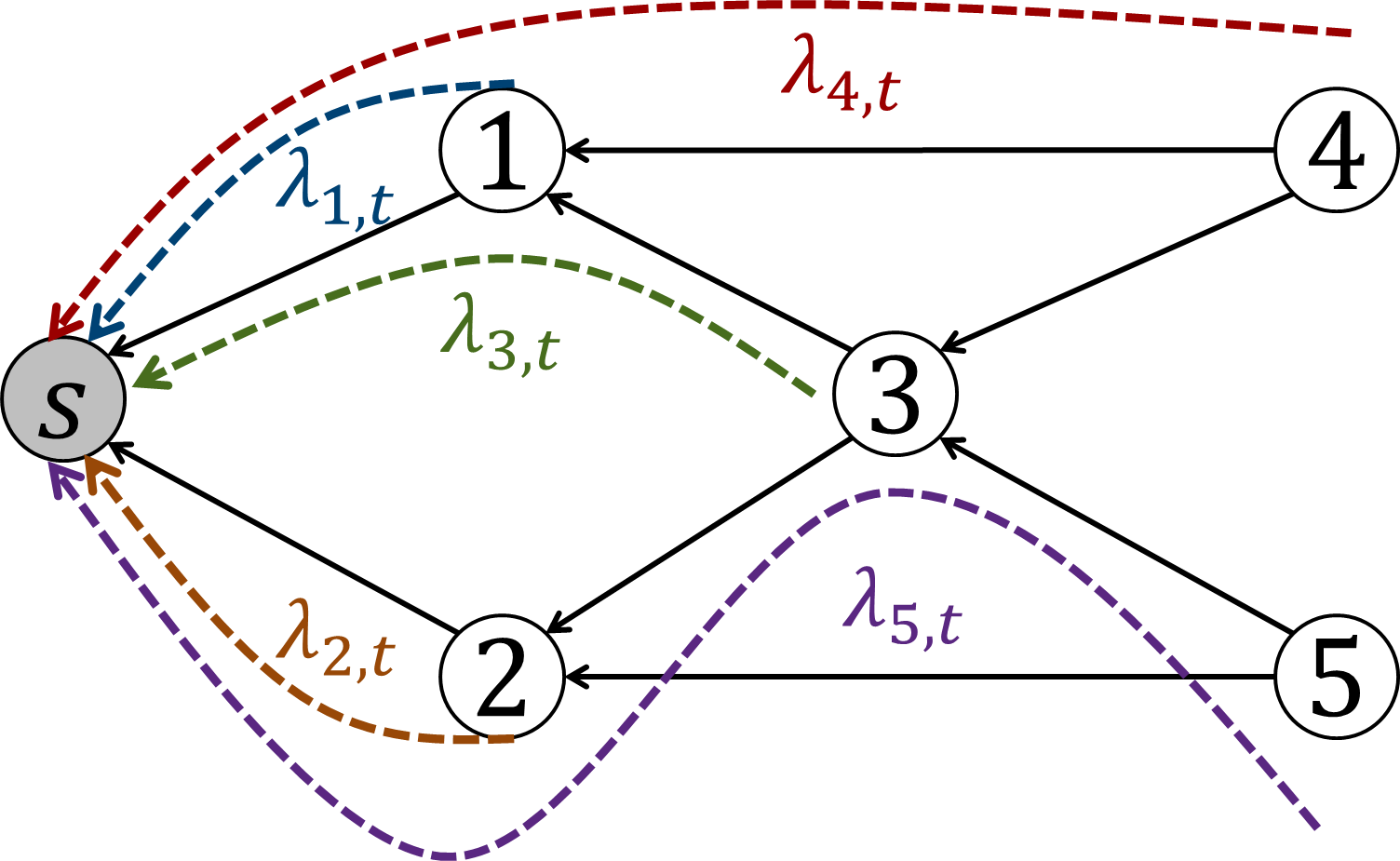}}\hspace{\fill}
\subfigure[Fractional routing.]{\label{fig:model--fractional}\includegraphics[scale = 0.14]{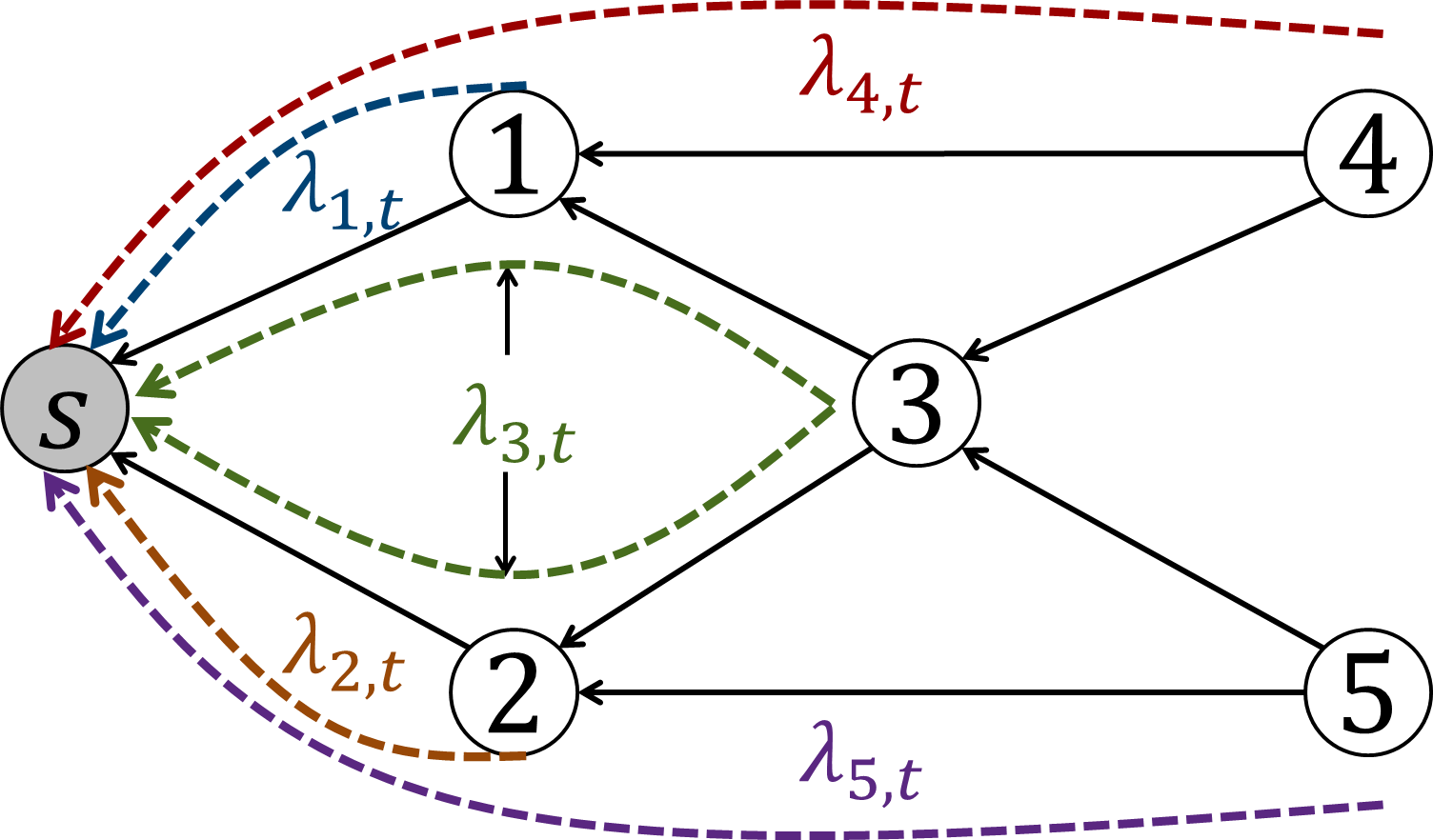}}\hspace*{\fill}
\caption{Routing types: \subref{fig:model--tree} a routing tree, \subref{fig:model--unsplittable} unsplittable routing: each node sends its data over one path, \subref{fig:model--fractional} fractional routing: nodes can send their data over multiple paths. Paths are represented by dashed lines.}
\label{fig:model}
\end{figure}
 
Each routing type incurs different trade-offs between the supported sensing rates\footnote{A metric of performance can be the minimum sensing rate that gets assigned to any node in any time slot.} and the required amount of control information. Routing types with higher number of active links require more control information to be exchanged between neighboring nodes (e.g., to maintain synchronization), and complicate the transmission and/or sleep-wake scheduling implementation. Moreover, energy consumed by the control messages can affect achievable rates significantly, due to limited energy budget, as confirmed via experiments in \cite{rob-experimental}. Below we outline the main characteristics of the routing types we consider.

\textbf{Routing Tree--}the simplest form of routing, in which every node sends all of the data it senses and receives to a single neighboring (parent) node. It requires minimum number of active links, yielding minimum energy consumption due to control messages. 
However, in general, it provides the lowest sensing rates (see more details below).

\textbf{Unsplittable Routing--}a single-path routing, in which every node sends all of its sensed data over a single path to the sink (a routing tree is a special case of the unsplittable routing, in which all the paths incoming into node $i$ outgo via the same edge). There are simple cases in which unsplitttable routing provides a rate assignment with the minimum sensing rate $\Omega(n)$ times higher than in a routing tree, where $n$ is the number of nodes\iffullresults{ (see Fig.~\ref{fig:tradeoffs--tree-unsplittable} for an example)}\fi. However, in general, it has higher number of active links than the routing tree, yielding higher energy consumption for control information.
\iffullresults
 \begin{figure*}[t]
\centering
\begin{minipage}[t]{.46\textwidth}
\centering
\includegraphics[scale=0.25]{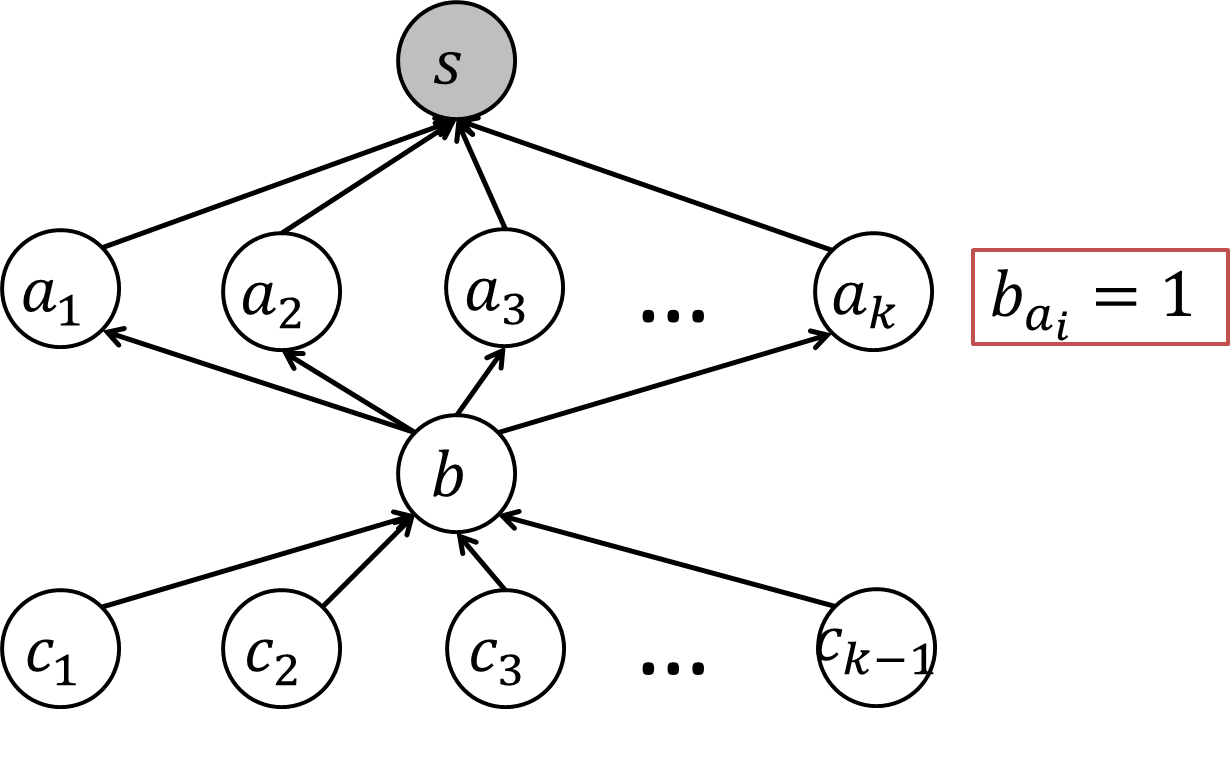}\vspace{-10pt}
\caption{A network example in which unsplittable routing provides minimum sensing rate that is $\Omega(n)$ higher than for any routing tree. Assume $c_{\text{st}}=c_{\text{rt}}=1$ and $T=1$. Available energies at all the nodes $a_i$, $i\in \{1,...,k\}$ are equal to 1, as shown in the box next to the nodes. Other nodes have energies that are high enough so that they are not constraining. In any routing tree, $b$ has some $a_i$ as its parent, so $\lambda_{a_i}=\lambda_b=\lambda_{c_1}=...=\lambda_{c_{k-1}}=1/(k+1)$ and $\lambda_j=1$ for $j\neq i$. In an unsplittable routing with paths $p_{a_i}=\{a_i, s\}$, $p_{c_i}=\{c_i, b, a_i, s\}$, and $p_{b}=\{b, a_k, s\}$, all the rates are equal to 1/2. As $k=\Theta(n)$, the minimum rate improves by $\Omega((k+1)/2)=\Omega(n)$.}\vspace{-10pt}
\label{fig:tradeoffs--tree-unsplittable}
\end{minipage}\hfill
\begin{minipage}[t]{.5\textwidth}
\centering
\includegraphics[scale=0.25]{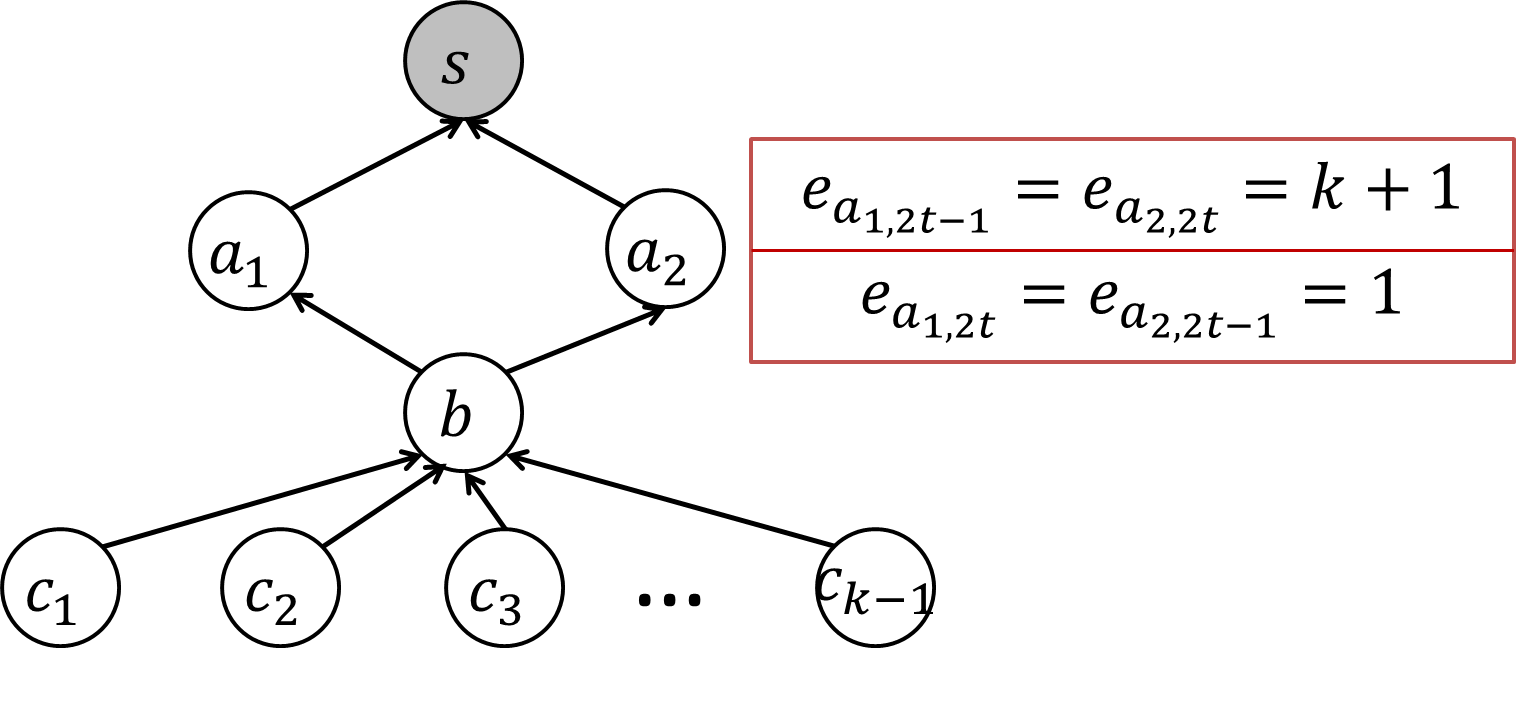}\vspace{-10pt}
\caption{A network example in which time-variable routing solution provides minimum sensing rate that is $\Omega(n)$ higher than for any time-invariable routing. The batteries of $a_1$ and $a_2$ are initially empty, and the battery capacity at all the nodes is $B=1$. Harvested energies over time slots for nodes $a_1$ and $a_2$ are shown in the box next to them. Other nodes are assumed not to be energy constraining. In any time-invariable routing, at least on of $a_1, a_2$ has $\Omega(k)=\Omega(n)$ descendants, forcing its rate to the value of $1/\Omega(n)$ in the slots in which the harvested energy is equal to 1. In a routing in which $b$ sends the data only through $a_1$ in odd slots and only through $a_2$ in even slots: $\lambda_b=\lambda_{c_1}=...=\lambda_{c_{k-1}}=1$.}\vspace{-10pt}
\label{fig:tradeoffs--time-variability}
\end{minipage}
\end{figure*}
\fi

\textbf{Fractional Routing--}a multi-path routing, in which each node can split its data over multiple paths to the sink (unsplittable routing is a special case of fractional routing in which every node has one path to the sink). It is the most general routing that subsumes both routing trees and unsplittable routings, and, therefore, provides the best sensing rates. However, it utilizes the highest number of links, yielding the highest energy consumption due to control messages.

\textbf{Time-invariable vs Time-variable Routing--} 
A routing is \emph{time-invariable} if every node uses the same (set of) path(s) in each time slot to send its data to the sink. If the paths change over time, the routing is \emph{time-variable}.\footnote{Whether the rates are constant or time-variable is independent of whether the routing is time-variable or not.} While there are cases in which time-variable routing provides a rate assignment with the minimum sensing rate $\Omega(n)$ times higher than in the time-invariable case\iffullresults{ (see, e.g., Fig.~\ref{fig:tradeoffs--time-variability})}\fi, it requires substantial control information exchange for routing reconfigurations, yielding high energy consumption.

\subsection{Our Contributions}

For the \emph{unsplittable routing} and \emph{routing tree}, we design a fully-combinatorial algorithm that solves the \emph{max-min fair rate assignment} problem, both in the time-variable and time-invariable settings, when the routing is provided at the input.
We then turn to \emph{fractional routing}, considering two settings: time-variable and time-invariable. We demonstrate that in the \emph{time-variable} setting verifying whether a given rate assignment is feasible is at least as hard as solving a feasible 2-commodity flow. This result implies that, to our current knowledge, it is unlikely that max-min fair time-variable fractional routing\footnote{We refer to a routing as max-min fair if it provides a lexicographically maximum rate assignment. The notions of max-min fairness and lexicographical ordering of vectors are defined in Section \ref{section:background}.} can be solved without the use of linear programming. To combat the high running time induced by the linear programming, we develop a fully polynomial time approximation scheme \hbox{(FPTAS)}. For the \emph{time-invariable setting}, we provide a fully-combinatorial algorithm that determines a max-min fair routing with constant rates.

Our rate assignment algorithms rely on the well-known water-filling framework, which is described in Section \ref{section:background}. It is important to note that water-filling is a framework--not an algorithm--and therefore it does not specify how to solve the maximization nor fixing of the rates steps (see Section \ref{section:background}). Even though general algorithms for implementing water-filling such as e.g., \cite{Radunovic2007, sensnet--lexicographic, OSU--lexicographic} can be adapted to solve the rate assignment problems studied in this paper, their implementation would involve solving $O(n^2T^2)$ linear programs (LPs) with $O(mT)$ variables and $O(nT)$ constraints, thus resulting in an unacceptably high running time. Moreover, such algorithms do not provide insights into the problem structure. Our algorithms are devised relying on the problem structure, and in most cases do not use linear programming. The only exception is the algorithm for time-variable fractional routing (Section \ref{section:fractional}), which solves $O(nT)$ LPs, and thus provides at least $O(nT)$-fold improvement as compared to an adaptation of \cite{Radunovic2007, sensnet--lexicographic, OSU--lexicographic}. Furthermore, each LP in our solution searches over much smaller space (in fact, only within the $\epsilon$-region of the starting point, for any $\epsilon$-approximation)

We show that \emph{determining} an unsplittable routing or a routing tree that supports lexicographically maximum rate assignment is NP-hard even for a single time slot. On the other hand, as a positive result, we develop an algorithm that determines a time-invariable unsplittable routing that maximizes the minimum sensing rate assigned to any node in any time slot.

The considered problems generalize classical max-min fair routing problems that have been studied outside the area of energy harvesting networks, such as: max-min fair fractional routing \cite{megiddo}, max-min fair unsplittable routing \cite{fairness-in-routing}, and bottleneck routing \cite{Bertsekas_Gallager}. In contrast to the problems studied in \cite{megiddo, fairness-in-routing, Bertsekas_Gallager}, our model allows different costs for flow generation and forwarding, and has time-variable node capacities determined by the available energies at the nodes. We note that studying networks with node capacities is as general as studying networks with capacitated edges, since there are standard methods for transforming one of these two problems into another (see, e.g., \cite{network_flows}). Therefore, we believe that the results can find applications in other related areas.

\subsection{Organization of the paper}
The rest of the \iffullresults paper \fi is organized as follows. Section \ref{section:system model} provides the model and problem formulations, which are placed in the context of related work in Section \ref{section:related work}. Section \ref{section:background} describes the connection between max-min fairness and lexicographic maximization. Section \ref{section:unsplittable} considers rate assignment in unsplittable routing, while Sections \ref{section:fractional} and \ref{section:fixed fractional} study fractional routing and rate assignment in time-variable and time-invariable settings, respectively. Section \ref{section:hardness} provides hardness results for determining unsplittable routing or a routing tree. \iffullresults Section \ref{section:conclusion} provides conclusions and outlines possible future directions. \fi

%% file: system-model.tex
We consider a network that consists of $n$ energy harvesting nodes and one sink node (see Fig.~\ref{fig:example network}). The sink is the central point at which all the sensed data is collected, and is assumed not to be energy constrained. In the rest of the paper, the term ``sink" will be used for the sink node. The connectivity between the nodes is modeled by a directed graph $G = (V, E)$, where $|V|=n+1$ ($n$ nodes and the sink), and $|E|=m$. We assume without loss of generality that every node has a directed path to the sink, because otherwise it can be removed from the graph.
 The main notation is summarized in Table \ref{table:notation}.

Each node is equipped with a rechargeable battery of finite capacity $B$. The time horizon is $T$ time slots. The duration of a time slot is assumed to be much longer than the duration of a single data packet, but short enough so that the rate of energy harvesting does not change during a slot. For example, if outdoor light energy is harvested, one time slot can be at the order of a minute. In a time slot $t$, a node $i$ harvests $e_{i, t}$ units of energy. The battery level of a node $i$ at the beginning of a time slot $t$ is $b_{i, t}$. We follow a predictable energy profile \cite{chen2011finite, maria-wiopt2011, maria-infocom2011, gurakan2013energy, shroff-srikant, OSU--lexicographic}, and assume that all the values of harvested energy $e_{i, t}$, $i \in \{1,..., n\}$, $t \in \{1,...,T\}$, battery capacity $B$, and all the initial battery levels $b_{i, 1}$, $i \in \{1,..., n\}$ are known and finite.

A node $i$ in slot $t$ senses data at rate $\lambda_{i, t}$. A node forwards all the data it senses and receives towards the sink. The flow on a link $(i, j)$ in slot $t$ is denoted by $f_{ij, t}$. Each node spends $c_{\text{s}}$ energy units to sense a unit flow, and $c_{\text{tx}}$, respectively $c_{\text{rx}}$, energy units to transmit, respectively receive, a unit flow.

\renewcommand{\arraystretch}{1.5}
\begin{table}[t!]
\small
\caption{Nomenclature.}
\centering
\renewcommand{\arraystretch}{0.8}
\begin{tabular}{|c|m{0.1\linewidth}| m{0.7\linewidth}|}
\hline
& & \\[-5pt]
\parbox[t]{2mm}{\multirow{9}{*}{\rotatebox[origin=c]{90}{\textbf{inputs}}}} & $n$ & Number of energy harvesting nodes\\
& $T$ & Time horizon\\
& $i$ & Node index, $i \in \{1, 2,...n\}$\\
& $t$ & Time index, $t \in \{1,..., T\}$\\
& $B$ & Battery capacity\\
& $e_{i, t}$ & Harvested energy at node $i$ in time slot $t$\\
& $c_{\text{s}}$ & Energy spent for sensing a unit flow\\
& $c_{\text{tx}}$ & Energy spent for transmitting a unit flow\\
& $c_{\text{rx}}$ & Energy spent for receiving a unit flow\\\hline
& & \\[-5pt]
\parbox[t]{2mm}{\multirow{3}{*}{\rotatebox[origin=c]{90}{\textbf{variables}}}} & &\\[-2pt]
& $\lambda_{i, t}$ & Sensing rate of node $i$ in time slot $t$\\
& $f_{ij, t}$ & Flow on link $(i, j)$ in time slot $t$\\
& $b_{i, t}$ & Battery level at node $i$ at the beginning of time slot $t$\\
& & \\
\hline
& & \\[-3pt]
\parbox[t]{2mm}{\multirow{2}{*}{\rotatebox[origin=c]{90}{\textbf{notation}}}} & & \\[-2pt]
& $c_{\text{st}}$ &Energy spent for jointly sensing and transmitting a unit flow: $c_{\text{st}} = c_{\text{s}}+c_{\text{tx}}$ \\
& $c_{\text{rt}}$ & Energy spent for jointly receiving and transmitting a unit flow: $c_{\text{rt}}= c_{\text{rx}}+c_{\text{tx}}$\\
& $f^{\Sigma}_{i, t}$ & Total flow entering node $i$ in time slot $t$: ${f^{\Sigma}_{i, t}= \littlesum_{j:(j, i)\in E}f_{ji, t}}$\\[8pt]
\hline
\end{tabular}\vspace{-10pt}
\label{table:notation}
\end{table}

The feasible region $\mathcal{R}$ for the sensing rates and flows is determined by the following set of linear{\footnote{{Note that we treat eq. (\ref{battery_evolution}) as a linear constraint, since the problems we are solving are maximizing $\lambda_{i, t}$'s (under the max-min fairness criterion), and (\ref{battery_evolution}) can be replaced by $b_{i, t+1}\leq B$ and $b_{i, t+1}\leq b_{i, t}+e_{i,t}- (c_{\text{rt}}f^{\Sigma}_{i, t} + c_{\text{st}}\lambda_{i, t})$ while leading to the same solution for $\{\lambda_{i, t}\}$.}}} constraints:
\begin{align}
\forall i\in \{1, ..., n\}, &t \in \{1,..., T\}:\notag\\
\littlesum_{(j, i)\in E}f_{ji, t} + \lambda_{i, t} &= \littlesum_{(i, j)\in E}f_{ij, t}\label{flow_conserv}\\
b_{i, t+1}=\min\{B,\quad b_{i, t}+&e_{i,t}- (c_{\text{rt}}f^{\Sigma}_{i, t} + c_{\text{st}}\lambda_{i, t})\}&\label{battery_evolution}\\
b_{i, t+1}\geq0, \>\lambda_{i, t}\geq 0,& \> f_{ij, t}\geq 0, \forall(i, j)\in E&\label{nonneg_const}
\end{align}
where $f^{\Sigma}_{i, t}\equiv\littlesum_{(j, i)\in E}f_{ji, t}$, $c_{\text{st}} \equiv c_{\text{s}}+c_{\text{tx}}$, and $c_{\text{rt}}\equiv c_{\text{rx}}+c_{\text{tx}}$.
Eq. (\ref{flow_conserv}) is a classical flow conservation constraint, while (\ref{battery_evolution}) models battery evolution over time slots.

Similar to the definition of max-min fairness in \cite{Bertsekas_Gallager}, define a rate assignment ${\{\lambda_{i, t}\}}$, $i \in \{1,...,n\}$, $t\in \{1, ..., T\}$, to be max-min fair if no $\lambda_{i, t}$ can be increased without either losing feasibility or decreasing some other rate $\lambda_{j, \tau}\leq \lambda_{i,t}$. Max-min fairness is closely related to lexicographic maximization, as will be explained in Section \ref{section:background}.

In some of the problems, the routing is provided at the input as a set of paths $\mathcal{P}=\{p_{i, t}\}$, for $i\in\{1,...,n\}, t\in\{1,...,T\}$. In such a case, $\mathcal{R}$ should be interpreted with respect to $\mathcal{P}$, instead with respect to the input graph $G$.

\subsection{Considered problems}  
We examine different routing types, in time-variable and time-invariable settings, as described in the introduction.
  For the unsplittable routing and routing trees, we examine the problems of \emph{determining a rate assignment} and \emph{determining a routing} separately, as described below. Table~\ref{table:notation} summarizes inputs and variables.

  \textsc{P-Unsplittable-Rates}: For a given time-variable unsplittable routing $\mathcal{P}=\{p_{i, t}\}$, determine the max-min fair assignment of the rates $\{\lambda_{i, t}\}$. Note that this setting subsumes time-invariable unsplittable routing, time-invariable routing tree, and time-variable routing tree.

 \textsc{P-Unsplittable-Find}: Associate with each (time-invari\-able or time-variable) unsplittable routing $\mathcal{P}$, a set of sensing rates $\{\lambda_{i, t}^{\mathcal{P}}\}$ that optimally solves \textsc{P-Unsplittable-Rates}. Determine an unsplittable routing $\mathcal{P}$ that provides the lexicographically maximum assignment of rates $\{\lambda_{i, t}^{\mathcal{P}}\}$.

  \textsc{P-Tree-Find}: Let $\mathcal{T}$ denote a (time-invariable or time-variable) routing tree on the input graph $G$. Associate with each $\mathcal{T}$ a set of sensing rates $\{\lambda_{i, t}^{\mathcal{T}}\}$ that optimally solves \hbox{\textsc{P-Unsplittable-Rates}}. Determine $\mathcal{T}$ that provides the lexicographically maximum assignment of rates $\{\lambda_{i, t}^{\mathcal{T}}\}$.

For the fractional routing, we study the following two variants of max-min fair routing, where the routing and the rate assignment are determined jointly.

 \textsc{{P-Fractional}}: Determine a \emph{time-variable} fractional routing that provides the max-min fair rate assignment $\{\lambda_{i, t}\}$, together with the max-min fair rate assignment $\{\lambda_{i, t}\}$, considering all the (time-variable, fractional) routings.

  \textsc{P-Fixed-Fractional}: Determine a \emph{time-invariable} fractional routing that provides the max-min fair time-invariable rate assignment $\{\lambda_{i, t}\}=\{\lambda_i\}$, together with the max-min fair rate assignment $\{\lambda_{i}\}$, considering all the (time-invariable, fractional) routings with rates that are constant over time.

%% file: related-work.tex
\noindent\textbf{Energy-harvesting Networks.}Rate assignment in energy harvesting networks in the case of a single node or a link was studied in \cite{ozel2011transmission, node-koksal, chen2011finite, maria-wiopt2011, maria-infocom2011, uysal2013finite, dohler-link}. 
{In \cite{chen2011finite}, the solution for one node is extended to the network setting by (1) formulating a convex optimization problem, and solving it using the Lagrange duality theory, and (2) using a heuristic algorithm. Both of these two solutions focus on maximizing the sum network throughput, and do not consider fairness.} 

  Resource allocation and scheduling for network-wide scenarios using the Lyapunov optimization technique was studied in \cite{gatzianas, neely, sarkar-kar, mao-koksal-shroff}. While the work in \cite{gatzianas, neely, sarkar-kar, mao-koksal-shroff} can support unpredictable energy profiles, it focuses on the (sum-utility of) time-average rates, which is, in general, time-unfair.  The design of online algorithms for resource allocation and routing was studied in \cite{shroff-srikant, chen2012simple}.

Max-min time-fair rate assignment for a single node or a link was studied in \cite{maria-wiopt2011, maria-infocom2011}, while max-min fair energy allocation for single-hop and two-hop scenarios was studied in \cite{gurakan2013energy}. Similar to our work, \cite{gurakan2013energy} requires fairness over both the nodes and the time slots, but considers only two energy harvesting nodes. 
The work on max-min fairness in network-wide scenarios \cite{OSU--lexicographic} is explained in more detail below.

\noindent\textbf{Sensor Networks.}
Problems \textsc{P-Fractional}, \textsc{P-Fixed-Fractional}, and \textsc{P-Unsplittable-Rates} are related to the maximum lifetime routing problems (see, e.g., \cite{chang2004maximum, madan2006distributed} and the follow-up work) in the following sense. 
In our model, maximization of the minimum sensing rate is equivalent to the network lifetime maximization in sensor networks, \emph{but only if the system is observed for} $T=1$. Namely, the nodes have the initial energy, and no harvesting happens over time.

Determining a maximum lifetime tree in sensor networks as in \cite{infocom2005} is a special case of \textsc{P-Tree-Find}. 
We extend the NP-hardness result from \cite{infocom2005} and provide a lower bound of $\Omega(\log n)$ for the approximation ratio (for both \cite{infocom2005} and \textsc{P-Tree-Find}), where $n$ is the number of nodes in the network.

\noindent\textbf{Max-min Fair Unsplittable Routing.} Rate assignment in unsplittable routing was studied extensively (see \cite{Bertsekas_Gallager, charny1995congestion} and references therein). \textsc{P-Unsplittable-Rates} reduces to the problem studied in \cite{Bertsekas_Gallager, charny1995congestion} for $c_{\text{st}}=c_{\text{rt}}$ and $T=1$. In the energy harvesting network setting, this problem has been studied in \cite{OSU--lexicographic}, for rates that are constant over time and a time-invariable routing tree. We consider a more general case than in \cite{OSU--lexicographic}, where the rates are time-variable, fairness is required over both network nodes and time slots, and the routing can be time-variable and given in a form of an unsplittable routing or a routing tree.

Determining a max-min fair unsplittable routing as studied in \cite{fairness-in-routing} is a special case of \textsc{P-Unsplittable-Find} for $c_{\text{st}}=c_{\text{rt}}$ and $T=1$, and the NP-hardness results from \cite{fairness-in-routing} implies the NP-hardness of \textsc{P-Unsplittable-Find}.

\noindent\textbf{Max-min Fair Fractional Routing.} Max-min fair fractional routing was first studied in \cite{megiddo}. The algorithm from \cite{megiddo} relies on the property that the total values of a max-min fair flow and max flow are equal, 
which does not hold even in simple instances of \textsc{P-Fixed-Fractional} and \textsc{P-Fractional}. \textsc{P-Fixed-Fractional} and \textsc{P-Fractional} reduce to the problem of \cite{megiddo} for $T=1$ and $c_{\text{st}}=c_{\text{rt}}$.

Max-min fair fractional routing in energy harvesting networks has been considered in \cite{OSU--lexicographic}. The distributed algorithm from \cite{OSU--lexicographic} solves \textsc{P-Fixed-Fractional}, but only as a heuristic. We provide a combinatorial algorithm that solves \textsc{P-Fixed-Fractional} optimally in a centralized manner.

A general linear programming framework for max-min fair routing was provided in \cite{Radunovic2007}, and extended to the setting of sensor and energy harvesting networks in \cite{sensnet--lexicographic} and \cite{OSU--lexicographic}, respectively. This framework, when applied to \textsc{P-Fractional}, is highly inefficient. \textsc{P-Fractional} reduces to \cite{sensnet--lexicographic} for $T=1$, and to \cite{OSU--lexicographic} when the rates are constant over time.

%% file: maxmin-lexmax-background.tex
Recall that a rate assignment ${\{\lambda_{i, t}\}}$, $i \in \{1,...,n\}$, $t\in \{1, ..., T\}$, is max-min fair if no rate $\lambda_{i, t}$ can be increased without either losing feasibility or decreasing some other rate $\lambda_{j, \tau}\leq \lambda_{i,t}$.
Closely related to the max-min fairness is the notion of lexicographic maximization. The lexicographic ordering of vectors, with the relational operators denoted by $\overset{lex}{=}$, $\overset{lex}{>}$, and $\overset{lex}{<}$, is defined as follows:
\begin{definition}
Let $u$ and $v$ be two vectors of the same length $l$, and let $u_s$ and $v_s$ denote the vectors obtained from $u$ and $v$ respectively by sorting their elements in the non-decreasing order. Then: (i) $u\overset{lex}{=}v$ if $u_s=v_s$ element-wise; (ii) $u\overset{lex}{>}v$ if there exists $j\in \{1, 2,..., l\}$, such that $u_s(j)>v_s(j)$, and $u_s(1)=v_s(1),...,u_s(j-1)=v_s(j-1)$ if $j>1$; (iii) $u\overset{lex}{<}v$ if not $u\overset{lex}{=}v$ nor $u\overset{lex}{>}v$.
\end{definition}

It has been proved in \cite{Radunovic2007} that a max-min fair allocation vector exists on any convex and compact set. The results from \cite{Sarkar-Tassiulas} state that in a given optimization problem whenever a max-min fair vector exists, it is unique and equal to the lexicographically maximum one.

In the problems \textsc{P-Unsplittable-Find}, \textsc{P-Fractional}, and \textsc{P-Fixed-Fractional} the feasible region $\mathcal{R}$ is determined by linear constraints (\ref{flow_conserv})-(\ref{nonneg_const}), and it is therefore convex. As we are assuming that all the input values $B, e_{i, t}$, and $b_{i, 1}$ are finite, it follows that the feasible region is also bounded, and therefore compact. Therefore, for the aforementioned problems, lexicographic maximization produces the max-min fair assignment of the sensing rates $\{\lambda_{i, t}\}$.

Lexicographic maximization can be implemented using the water-filling framework (see, e.g., \cite{Bertsekas_Gallager}):
\begin{algorithm}
\caption{\textsc{Water-filling-Framework}($G, b, e$)}
 \begin{algorithmic}[1]
\State Set $\lambda_{i, t} = 0$ $\forall i, t$, and mark all the rates as not fixed.
\State  Increase all the rates $\lambda_{i, t}$ that are not fixed by the same maximum amount, subject to the constraints from $\mathcal{R}$. \label{maximizing}
\State  Fix all the $\lambda_{i, t}$\rq{}s that cannot be further increased.\label{fixing}
\State  If all the rates are fixed, terminate. Else, go to step 2.
\end{algorithmic}
\end{algorithm}

To solve the problems \textsc{P-Fractional}, \textsc{P-Fixed-Fractio\-nal}, and \textsc{P-Unsplittable-Rates}, the challenge is to implement the Steps \ref{maximizing} and \ref{fixing} efficiently, which we study in the following sections. We will refer to the algorithm that implements Step \ref{maximizing} as \textsc{Maximizing-the-Rates}, and to the algorithm that implements Step \ref{fixing} as \textsc{Fixing-the-Rates}.

\begin{note} A rate $\lambda_{i, t}$ can in general get fixed in any iteration of the \textsc{Water-filling-Framework}; there is no rule that relates an iteration $k$ to a node $i$ or a time slot $t$.
\end{note}

%% file: unsplittable-rates.tex
This section studies \textsc{P-Unsplittable-Rates}, the problem of rate assignment for an unsplittable routing provided at the input. The analysis applies to any time-invariable or time-variable unsplittable routing or a routing tree.

We assume that the routing over time $t\in\{1,...,T\}$ is provided as a set of routing paths $\mathcal{P}=\{p_{i, t}\}$ from a node $i$ to the sink $s$, for each node $i\in V\backslash\{s\}$. We say that a node $j$ is a descendant of a node $i$ in a time slot $t$ if $i \in p_{j, t}$. \footnote{Notice that this is consistent with the definition of a descendant in a routing tree.} 

As observed in Section \ref{section:background}, to design an efficient rate assignment algorithm, we need to implement the Steps \ref{maximizing} and \ref{fixing} of \textsc{Water-filling-Framework} efficiently. 

Before describing the algorithms in detail, we need to introduce some notation. Let $F_{i, t}^k=1$ if the rate $\lambda_{i, t}$ is not fixed at the beginning of the $k^{\text{th}}$ iteration of \textsc{Water-filling-Framework}, $F_{i, t}^k=0$ otherwise. Initially, $F_{i, t}^1=1$, $\forall i, t$. If a rate $\lambda_{i, t}$ is not fixed, we will say that it is \lq\lq{}active\rq\rq{}. We will denote by $D_{i, t}^k$ the number of active descendants of the node $i$ in the time slot $t$, where $D_{i, t}^1=|\{j: i \in p_{j, t}\backslash\{j\}\}|$. Notice that $D_{i, t}^k=\littlesum_{j : i\in p_{j, t}\backslash\{j\}}F_{j, t}^k $. Finally, let $\lambda_{i, t}^k$ denote the value of $\lambda_{i, t}$ in the $k^{\text{th}}$ iteration of \textsc{Water-filling-Framework}, and let $\lambda_{i, t}^0=0$, $\forall i,t$. Under this notation, the rates can be expressed as
$
\lambda_{i, t}^k=\littlesum_{l=1}^k F_{i, t}^l\lambda^l,
$
where $\lambda^l$ denotes the common amount by which all the active rates get increased in the $l^{\text{th}}$ iteration.

\subsection{Maximizing the Rates}\label{sec-maximizing-rates}

Maximization of the common rate $\lambda^k$ in $k^{\text{th}}$ iteration of \textsc{Water-filling-Framework} can be formulated as follows:
\begin{align*}
\textbf{max } \>&\lambda^k\\ 
\textbf{s.t. } &\forall i\in \{1, ..., n\}, t \in \{1,..., T\}:\\
&\lambda_{i, t}^{k} = \lambda_{i, t}^{k-1} + F_{i, t}^k\cdot \lambda^k&\\
&f_{i, t}^{\Sigma} + \lambda_{i, t}^k = \littlesum_{(i, j)\in E}f_{ij, t}\\
&b_{i, t+1}=\min\{ B,\> b_{i, t}+e_{i, t}- (c_{\text{rt}}f_{i, t}^{\Sigma} +c_{\text{st}}\lambda_{i, t}^k )\}&\\
&b_{i, t}\geq0, \>\lambda^k \geq 0, \> f_{ij, t}\geq 0, \forall (i, j)\in E
\end{align*}

As in each slot $t$ every node $i$ sends all the flow it senses over a single path, we can compute the total inflow into a node $i$ as the sum of the flows coming from $i$'s descendants:
\begin{align*}
f_{i, t}^{\Sigma}&=\littlesum_{j : i\in p_{j, t}\backslash\{j\}}\littlesum_{l =1}^k F_{j, t}^l\cdot\lambda^l = \littlesum_{l =1}^k \lambda^l \littlesum_{j : i\in p_{j, t}\backslash\{j\}}F_{j, t}^l \\
&= \littlesum_{l =1}^k D_{i, t}^l\cdot \lambda^l
\end{align*}

Denoting the battery levels in the iteration $k$ as $b_{i, t}^k$, the problem can now be written more compactly as:
\begin{align*}
\textbf{max } \>&\lambda^k\\ 
\textbf{s.t. } &\forall i\in \{1, ..., n\}, t \in \{1,..., T\}:\\
&b_{i, t+1}^k = \min\{B, \; b_{i, t}^k+e_{i, t}- \littlesum_{l=1}^k \lambda^l(c_{\text{rt}}D_{i, t}^l + c_{\text{st}}F_{i, t}^l )\}\label{one-line-battery}\\
&b_{i, t}^k\geq0, \>\lambda^k \geq 0,
\end{align*}
where $\forall i \,\forall k:$ $b_{i, 1}^k=b_{i, 1}$.

\iffullresults
Instead of using all of the $\lambda^l$'s from previous iterations in the expression for $b_{i, t+1}^k$, we can define the battery drop in the iteration $k$, for node $i$ and time slot $t$ as: $\Delta b_{i, t}^k = \littlesum_{l=1}^k \lambda^l\left(c_{\text{rt}}D_{i, t}^l + c_{\text{st}}F_{i, t}^l \right)$ and only keep track of the battery drops from the previous iteration. The intuition is as follows: to determine the battery levels in all the time slots, we only need to know the initial battery level and how much energy ($\Delta b_{i, t}$) is spent per time slot. Setting $\Delta b_{i, t}^0 = 0$, the problem can be written as: 
\fi
\begin{align*}
\textbf{max } \>&\lambda^k\\ 
\textbf{s.t. } &\forall i\in \{1, ..., n\}, t \in \{1,..., T\}:\\
&\Delta b_{i, t}^k = \Delta b_{i, t}^{k-1} +\lambda^k(c_{\text{rt}}D_{i, t}^k + c_{\text{st}}F_{i, t}^k )\\
&b_{i, t+1}^k = \min\{B, \> b_{i, t}^k+e_{i, t}- \Delta b_{i, t}^{k}\}\\
&b_{i, t}^k\geq0, \>\lambda^k \geq 0
\end{align*}

Writing the problem for each node independently, we can solve the following subproblem:
\begin{flalign}
\textbf{max } \>&\lambda_i^k\\ 
\textbf{s.t. } &\forall t \in \{1,..., T\}:\notag\\
&\Delta b_{i, t}^k = \Delta b_{i, t}^{k-1} +\lambda_i^k(c_{\text{rt}}D_{i, t}^k + c_{\text{st}}F_{i, t}^k )\label{delta-b}\\
&b_{i, t+1}^k = \min\{B, \> b_{i, t}^k+e_{i, t}- \Delta b_{i, t}^{k}\}\label{battery-with-delta-b}\\
&b_{i, t}^k\geq0, \>\lambda_i^k \geq 0
\end{flalign}
for each $i$ with $\littlesum_{i, t}F^k_{i, t}>0$, 
and determine $\lambda^k=\min_{i}\lambda_i^k$. 
Notice that we can bound each $\lambda_i^k$ by the interval $[0, \lambda_{\max,i}^k]$, where $\lambda_{\max, i}^k$ is the rate for which node $i$ spends all its available energy in the first slot $\tau$ in which its rate is not fixed:
\begin{equation*}
\lambda_{\max, i}^k = \frac{b_{i, \tau}^{k-1}+ e_{i, \tau}}{c_{\text{rt}}D_{i, \tau}^k + c_{\text{st}}},\quad \tau = \min\{t: F_{i, t}^k = 1\}.
\end{equation*}

The subproblem of determining $\lambda_i^k$ can now be solved by performing a binary search in the interval $[0, \lambda_{\max,i}^k]$. 

Let $\delta$ denote the precision of the input variables. Note that however small, $\delta$ can usually be expressed as a constant. 
\begin{lemma}\label{lemma:unsplittable-max-rates-time} 
\textsc{Maximizing-the-Rates} in \textsc{P-Unsplittable-Find} can be implemented in time \begin{equation*}O\left(T\littlesum_{i}\log\left({\lambda_{\max, i}^k}/{\delta}\right)\right) = O\left(nT\log\left( {B+\max_{i, t} e_{i, t}}/({\delta c_{\text{st}}}) \right)\right).\end{equation*}
\end{lemma}

\subsection{Fixing the rates}\label{sec-fixing-rates}

Recall that the elements of the matrix $F^k$ are such that $F_{i, t}^k=0$ if the rate $\lambda_{i, t}$ is fixed for the iteration $k$, and $F_{i, t}^k=1$ otherwise. At the end of iteration $k\geq 1$, let $F^{k+1} = F^{k}$, and consider the following set of rules for fixing the rates:
\begin{enumerate}[itemsep=-3pt, label=(F\arabic*), topsep=3pt]
\item For all $(i, t)$ such that $b_{i, t+1}^k=0$ set $F_{i, t}^{k+1}=0$. \label{spent-all-energy}
\item \label{previous-energy}For all $(i, t)$ such that $b_{i, t+1}^k=0$ determine the longest sequence $(i, t), (i, t-1), (i, t-2),..., (i, \tau), \tau\geq 1,$ with the property that $b_{i, s}^k+e_{i, s}- \Delta b_{i, s}^{k}\leq B$ $\forall s\in\{t, t-1,..., \tau\}$, and set $F_{i, s}^{k+1}=0$ $\forall s$. 
\item For all $(i, t)$ for which the rules \ref{spent-all-energy} and \ref{previous-energy} have set $F_{i, t}^{k+1}=0$, and for all $j$ such that $i\in p_{j, t}$, set $F_{j, t}^{k+1}=0$. \label{descendants}
\end{enumerate}

\iffullresults
We will need to prove that these rules are necessary and sufficient for fixing the rates. Here, \lq\lq{}necessary\rq\rq{} means that no rate that gets fixed at the end of iteration $k$ can get increased in iteration $k+1$ without violating at least one of the constraints. \lq\lq{}Sufficient\rq\rq{} means that all the rates $\lambda_{i, t}$ with $F_{i, t}^{k+1}=1$ can be increased by a positive amount in iteration $k+1$ without violating feasibility.
\fi
\begin{lemma}\label{necessity}
(Necessity) No rate fixed by the rules \ref{spent-all-energy}, \ref{previous-energy} and \ref{descendants} can be increased in the next iteration without violating feasibility constraints.
\end{lemma}
\iffullresults
\begin{proof}
We will prove the lemma by induction. 

\textbf{The base case.} Consider the first iteration and observe the pairs $(i, t)$ for which $F^1_{i, t}=0$.

Suppose that $b_{i, t+1}^1 = 0$. The first iteration starts with all the rates being active, so we get from the constraint (\ref{battery-with-delta-b}):  
\begin{align}
b_{i, t+1}^1&=\min\left\{B, \> b_{i, t}^1+ e_{i, t}-  \left(c_{\text{rt}}\littlesum_{i\in p_{j, t}\backslash\{j\}}\lambda_{j, t}^1 + c_{\text{st}}\lambda_{i, t}^1\right)\right\}\notag\\
&=  b_{i, t}^1+e_{i, t}- \left(c_{\text{rt}}\littlesum_{i\in p_{j, t}\backslash\{j\}}\lambda_{j, t}^1 + c_{\text{st}}\lambda_{i, t}^1\right)=0\label{th1-all-energy-spent}
\end{align} 
as $B>0$, where the third line comes from all the rates being equal in the first iteration and the fact that all the $i$\rq{}s descendants must send their flow through $i$.

As every iteration only increases the rates, if we allow $\lambda_{i, t}$ to be increased in the next iteration, then (from (\ref{th1-all-energy-spent})) we would get $b_{i, t+1}<0$, which is a contradiction. Alternatively, if we increase $\lambda_{i, t}^1$ at the expense of decreasing some $\lambda_{j, t}^1$, $i\in p_{j, t}\backslash\{j\}$, to keep $b_{i, t+1}\geq 0$, then the solution is not max-min fair, as $\lambda_{j, t}^1 = \lambda_{i, t}^1 = \lambda^1$. This proves the necessity of the rule \ref{spent-all-energy}. By the same observation, if we increase the rate $\lambda_{j, t}^1$ of any of the node $i$'s descendants $j$ at time $t$, we will necessarily get $b_{i, t+}<0$ (or would need to sacrifice the max-min fairness). This proves the rule \ref{descendants} for all the descendants of node $i$, such that $F_{i, t}^{2}$ is set to 0 by the rule \ref{spent-all-energy}. 

Now let $(i, t), (i, t-1), (i, t-2),..., (i, \tau), \tau\geq 1,$ be the longest sequence with the property that: $b_{i, t}=0$ and $b_{i, s}^1+e_{i, s}- \Delta b_{i, s}^{1}\leq B$ $\forall s\in\{t, t-1,..., \tau\}$. Observe that when this is the case, we have: 
\begin{align*}
\forall s \in&\{\tau, \tau+1,..., t-2, t-1\}:\\
b_{i, s+1}^1&=\min\left\{B,b_{i, s}^1+e_{i, s}- \Delta b_{i, s}^{1}\right\}=  b_{i, s}^1+e_{i, s}- \Delta b_{i, s}^{1}\notag\\
&= b_{i, s}^1+e_{i, s}- \left(c_{\text{rt}}\littlesum_{j: i\in p_{j, t}\backslash\{j\}}\lambda_{j, s}^1 + c_{\text{st}}\lambda_{i, s}^1\right)
\end{align*}
This gives a recursive relation, so $b_{i, t+1}$ can also be written as:
\begin{equation*}
b_{i, t+1}^1 = b_{i, \tau}^1 + \littlesum_{s=\tau}^t e_{i, s} - c_{\text{rt}}\littlesum_{s=\tau}^t\littlesum_{j: i\in p_{j, t}\backslash\{j\}}\lambda_{j, s}^1 - c_{\text{st}}\littlesum_{s=\tau}^t \lambda_{i, s}^1.
\end{equation*}
If we increase $\lambda_{i, s}$ or $\lambda_{j, s}$, for any $j, s$ such that $i\in p_{j, s}\backslash\{j\}$, $s \in \{\tau, \tau+1,..., t-2, t-1\}$, then either $b_{i, t+1}$ becomes negative, or we sacrifice the max-min fairness, as all the rates are equal to $\lambda^1$ in the first iteration. This proves both the rule \ref{previous-energy} and completes the proof for the necessity of the rule \ref{descendants}.

\textbf{The inductive step.} Suppose that all the rules are necessary for the iterations $1,2,...k-1$, and consider the iteration $k$.

Observe that:
\begin{enumerate}[itemsep=-3pt, label=(o\arabic*), , topsep=3pt]
\item $\lambda_{j, t}\leq \lambda_{i, t}$, $\forall j: i\in p_{j, t}$, as all the rates, until they are fixed, get increased by the same amount in each iteration, and once a rate gets fixed for some $(i, t)$, by the rule \ref{descendants}, it gets fixed for all the node $i$'s descendants in the same time slot. Notice that the inequality is strict only if $\lambda_{j, t}$ got fixed before $\lambda_{i, t}$; otherwise these two rates get fixed to the same value.\label{o2}
\item Once fixed, a rate never becomes active again.\label{o3}
\item If a rate $\lambda_{i, t}$ gets fixed in iteration $k$, then $\lambda_{i, t}= \lambda_{i, t}^k = \littlesum_{p=1}^k \lambda^p=\lambda_{i, t}^l $, $\forall l>k$.\label{o4}
\end{enumerate} 
Suppose that $b_{i, t+1}^k = 0$ for some $i\in\{1,..,n\}, t\in\{1,...,T\}$. If $F_{i, t}^{k}=0$, then by the inductive hypothesis $\lambda_{i, t}$ cannot be further increased in any of the iterations $k, k+1,...$. Assume $F_{i, t}^{k}=1$. 
Then:
\begin{align*}
b_{i, t+1}^k&=\min\left\{B, \> b_{i, t}^k+e_{i, t}- \left(c_{\text{rt}}\littlesum_{j: i\in p_{j, t}\backslash\{j\}}\lambda_{j, t}^k + c_{\text{st}}\lambda_{i, t}^k \right)\right\}\\
&= b_{i, t}^k+e_{i, t}- \left(c_{\text{rt}}\littlesum_{j: i\in p_{j, t}\backslash\{j\}}\lambda_{j, t}^k + c_{\text{st}}\lambda_{i, t}^k \right)=0
\end{align*}
By the observation \ref{o2}, $\lambda_{j, t}^k\leq \lambda_{i, t}^k$, $\forall j$ such that $i \in p_{j, t}\backslash\{j\}$, where the inequality holds with equality if $F_{j, t}^{k}=0$. Therefore, if we increase $\lambda_{i, t}$ in some of the future iterations, either $b_{i, t+1}<0$, or we need to decrease some $\lambda_{j, t}\leq \lambda_{i, t}$, violating the max-min fairness condition. This proves the necessity of the rule \ref{spent-all-energy}. For the rule \ref{descendants}, as for all $(j, t)$ with $F_{j, t}^{k}=1, \, i \in p_{j, t}\backslash\{j\}$, we have $\lambda_{j, t}= \lambda_{i, t}$, none of the $i$'s descendants can further increase its rate in the slot $t$.

Now for $(i, t)$ such that $b_{i, t+1}^k = 0$, let $(i, t), (i, t-1), (i, t-2),..., (i, \tau), \tau\geq 1,$ be the longest sequence with the property that: $b_{i, s}^k+e_{i, s}- \Delta b_{i, s}^{k}\leq B$ $\forall s\in\{t, t-1,..., \tau\}$. Similarly as for the base case:
\begin{align*}
\forall s \in&\{\tau, \tau+1,..., t-2, t-1\}:\\
b_{i, s+1}^k&=\min\left\{B, b_{i, s}^k+e_{i, s}- \Delta b_{i, s}^{k}\right\}\\
&= b_{i, s}^k+e_{i, s}- \left(c_{\text{rt}}\littlesum_{j: i\in p_{j, t}\backslash\{j\}}\lambda_{j, s}^k + c_{\text{st}}\lambda_{i, s}^k\right) 
\end{align*}
and we get that:
\begin{equation}
b_{i, t+1}^k = b_{i, \tau}^k + \littlesum_{s=\tau}^t e_{i, s} - c_{\text{rt}}\littlesum_{s=\tau}^t\littlesum_{j: i\in p_{j, t}\backslash\{j\}}\lambda_{j, s}^k - c_{\text{st}}\littlesum_{s=\tau}^t \lambda_{i, s}^k.\label{rec-battery}
\end{equation}
If any of the rates appearing in (\ref{rec-battery}), was fixed in some previous iteration, then it cannot be further increased by the inductive hypothesis. By the observation \ref{o2}, all the rates that are active are equal, and all the rates that are fixed are strictly lower than the active rates. Therefore, by increasing any of the active rates from (\ref{rec-battery}), we either violate battery nonnegativity constraint or the max-min fairness condition. Therefore, rule \ref{previous-energy} holds, and rule \ref{descendants} holds for all the descendants of nodes whose rates got fixed by the rule \ref{previous-energy}, in the corresponding time slots. 
\end{proof}
\fi

\begin{lemma}\label{sufficiency}
(Sufficiency) If $F_{i, t}^{k+1}=1$, then $\lambda_{i, t}$ can be further increased by a positive amount in the iteration $k+1$, $\forall i \in \{1,..., n\}, \, \forall t \in\{1,...,T\}$. 
\end{lemma}
\begin{proof}
Suppose that $F_{i, t}^{k+1}=1$. Notice that by increasing $\lambda_{i, t}$ by some $\Delta \lambda_{i, t}$ node $i$ spends additional $\Delta b_{i, t}=c_{\text{st}}\Delta \lambda_{i, t}$ energy \emph{only in the time slot t}. As $F_{i, t}^{k+1}=1$, by the rules \ref{spent-all-energy} and \ref{previous-energy}, either $b_{i, t'}>0$ $\forall t'>t$, or there is a time slot $t''>t$ such that $b_{i, t''}^k+e_{i, t''}- \Delta b_{i, t''}^{k}>B$ and $t''<t'''$, where $t'''=\arg\min\left\{\tau>t: b_{i, \tau}=0\right\}$. 

If $b_{i, t'}>0$ $\forall t'>t$, then the node $i$ can spend $\Delta b_{i, t} = \min_{t+1\leq t'\leq T+1} b_{i, t'}^k$ energy, and keep $b_{i, t'}\geq 0$, $\forall t'$, which follows from the battery evolution equation (\ref{battery-with-delta-b}). 

If there is a slot $t'''>t$ in which $b_{i, t'''}^k=0$, then let $t''$ be the minimum time slot between $t$ and $t'''$, such that $b_{i, t''}^k+e_{i, t''}- \Delta b_{i, t''}^{k}>B$. Decreasing the battery level at $t''$ by $(b_{i, t''}^k+e_{i, t''}- \Delta b_{i, t''}^{k})-B$ does not influence any other battery levels, as in either case $b_{i, t''+1} = B$. As all the battery levels are positive in all the time slots between $t$ and $t''$, $i$ can spend at least $\min\{(b_{i, t''}^k+e_{i, t''}- \Delta b_{i, t''}^{k})-B, \quad\min_{t+1\leq t'\leq t''} b_{i, t'}^k\}$ energy at the time $t$ and have $b_{i, t'}\geq0$ $\forall t'$.

By the rule \ref{descendants}, $\forall j$ such that $j\in p_{i, t}$ we have that $b_{j, t}>0$, and, furthermore, if $\exists t'''>t$ with $b_{j, t'''}=0$ then $\exists t''\in\{t, t'''\}$ such that $b_{i, t''}^k+e_{i, t''}- \Delta b_{i, t''}^{k}>B$. By the same observations as for the node $i$, each $j\in p_{i, t}$ can spend some extra energy $\Delta b_{j, t}>0$ in the time slot $t$ and keep all the battery levels nonnegative. In other words, there is a directed path from the node $i$ to the sink on which every node can spend some extra energy in time slot $t$ and keep its battery levels nonnegative. Therefore, if we keep all other rates fixed, the rate $\lambda_{i, t}$ can be increased by $\Delta \lambda_{i, t} =\min\{ {\Delta b_{i, t}}/{c_{\text{st}}}, \min_{j\in p_{i, t}} {\Delta b_{j, t}}/{c_{\text{rt}}}\}>0$. 

As each active rate $\lambda_{i, t}$ can (alone) get increased in the iteration $k+1$ by some $\Delta \lambda_{i, t}>0$, it follows that all the active rates can be increased simultaneously by at least\\ $\min_{i, t}\Delta\lambda_{i, t}/(T(c_{\text{st}}+nc_{\text{rt}}))>0$.
\end{proof}

\begin{theorem}\label{thm:unsplittable-fixing}
Fixing rules \ref{spent-all-energy}, \ref{previous-energy} and \ref{descendants} provide necessary and sufficient conditions for fixing the sensing rates in \textsc{Water-filling-Framework}.
\end{theorem}\iffullresults
\begin{proof}
Follows directly from Lemmas \ref{necessity} and \ref{sufficiency}.
\end{proof}
\fi

\begin{lemma}\label{lemma:unsplittable-fixing-time}
The total running time of \textsc{Fixing-the-Rates} in \textsc{P-Unsplittable-Find} is $O(mT)$.
\end{lemma}
\iffullresults 
\begin{proof}
Rules \ref{spent-all-energy} and \ref{previous-energy} can be implemented for each node independently in time $O(T)$ by examining the battery levels from slot $T+1$ to slot $2$. 

For the rule \ref{descendants}, in each time slot $t\in\{1,...,T\}$ enqueue all the nodes $i$ whose rates got fixed in time slot $t$ by either of the rules \ref{spent-all-energy}, \ref{previous-energy} and perform a breadth-first search. Fix the rates of all the nodes discovered by a breadth-first search. This gives $O(m)$ time per slot, for a total time of $O(mT)$. Combining with the time for rules \ref{spent-all-energy} and \ref{previous-energy}, the result follows.
\end{proof}\fi

Combining Lemmas \ref{lemma:unsplittable-max-rates-time} and \ref{lemma:unsplittable-fixing-time}, we can compute the total running time of \textsc{Water-filling-Framework} for \textsc{P-Unsplit\-table-Find}, as stated in the following lemma.

 \begin{lemma}\label{lemma:unsplittable-total-time}
 \textsc{Water-filling-Framework} with Steps \ref{maximizing} \textsc{Maximizing-the-Rates} and \ref{fixing} \textsc{Fixing-the-Rates} implemented as described in Section \ref{section:unsplittable} runs in time: \begin{equation}O(nT(mT+nT\log({B+\max_{i, t} e_{i, t}}/{(\delta c_{\text{st}})} ))).\notag\end{equation}
 \end{lemma}
 \iffullresults
 \begin{proof}
 To bound the running time of the overall algorithm that performs lexicographic maximization, we need to first bound the number of iterations that the algorithm performs. 
As in each iteration at least one sensing rate $\lambda_{i, t}$, $i \in \{1,..., n\}$, $t\in\{1,...,T\}$, gets fixed, and once fixed remains fixed, the total number of iterations is $O(nT)$. The running time of each iteration is determined by the running times of the steps \ref{maximizing} (\textsc{Maximizing-the-Rates}) and \ref{fixing} (\textsc{Fixing-the-Rates}) of the \textsc{Water-filling-Framework}. \textsc{Maximizing-the-Rates} runs in $O\left(nT\log\left( \frac{B+\max_{i, t} e_{i, t}}{\delta c_{\text{st}}} \right)\right)$ (Lemma \ref{lemma:unsplittable-max-rates-time}), whereas \textsc{Fixing-the-Rates} runs in $O(mT)$ time (Lemma \ref{lemma:unsplittable-fixing-time}). Therefore, the total running time is: $O\left(nmT^2+n^2T^2\log\left({B+\max_{i, t} e_{i, t}}/{(\delta c_{\text{st}})} \right)\right)$.
 \end{proof}
 \fi

%% file: fractional.tex
The feasible region $\mathcal{R}$ for the rates and flows in \textsc{P-Fractio\-nal} can be described by the following constraints:
\begin{align*}
&\forall i\in \{1, ..., n\}, t \in \{1,..., T\}:\\
&f^{\Sigma}_{i, t} + \lambda_{i, t} = \littlesum_{(i, j)\in E}f_{ij, t}\\
&b_{i, t+1}=\min\{ B,\> b_{i, t}+e_{i, t}
- (c_{\text{rt}}f^{\Sigma}_{i, t}+ c_{\text{st}}\lambda_{i, t} )\}&\\
&b_{i, t}\geq0, \>\lambda_{i, t} \geq 0, \> f_{ij, t}\geq 0, \forall (i, j)\in E,
\end{align*}
where $f^{\Sigma}_{i, t}\equiv \littlesum_{(j, i)\in E}f_{ji, t}$.

Observe that we can avoid computing the values of battery levels $b_{i, t+1}$, and instead explicitly write the non-negativity constraints for each of the terms inside the $\min$. Reordering the terms, we get the following formulation:
\begin{flalign}
&\forall i\in \{1, ..., n\}, t \in \{1,..., T\}:\notag\\
&f^{\Sigma}_{i, t} + \lambda_{i, t} = \littlesum_{(i, j)\in E}f_{ij, t}\label{packing-flow-balance}\\
&\littlesum_{\tau=1}^t (c_{\text{rt}}f^{\Sigma}_{i, \tau} + c_{\text{st}}\lambda_{i, t}) \leq b_{i, 1}+ \littlesum_{\tau=1}^t e_{i, \tau}\label{packing-battery-not-full}\\
&\littlesum_{\tau=s}^t (c_{\text{rt}}f^{\Sigma}_{i, \tau} + c_{\text{st}}\lambda_{i, t}) \leq B+\littlesum_{\tau=s}^t e_{i, \tau}, \quad2\leq s\leq t \label{packing-battery-full}\\
&\lambda_{i, t} \geq 0, \> f_{ij, t}\geq 0, \forall (i, j)\in E\label{packing-nonnegativity}
\end{flalign}

In the $k^{\text{th}}$ iteration of \textsc{Water-filling-Framework} we have that $\lambda_{i, t}^k=\lambda_{i, t}^{k-1}+F_{i, t}^k\cdot \lambda^k=\littlesum_{l=1}^k F_{i, t}^l\cdot\lambda^l$, where $\lambda_{i, t}^0=0$. Let:
\begin{equation*}
u_{i, t}^b=b_{i, 1}+ \littlesum_{\tau=1}^t(e_{i, \tau}- c_{\text{st}}\lambda_{i, \tau}^{k-1} ),\;
u_{i, t, s}^B=B+\littlesum_{\tau=s}^t(e_{i, \tau}- c_{\text{st}}\lambda_{i, \tau}^{k-1} )
\end{equation*}
Since in the iteration $k$ all $\lambda_{i, t}^{k-1}$\rq{}s are constants, the rate maximization subproblem can be written as:
\begin{flalign}
\textbf{max } \>&\lambda^k\label{packing-single}\\
\textbf{s.t. }  &\forall i\in \{1, ..., n\}, t \in \{1,..., T\}:\notag\\
&-f^{\Sigma}_{i, t} - F_{i, t}^k\cdot \lambda^k  + \littlesum_{(i, j)\in E}f_{ij, t} = \lambda_{i, t}^{k-1}\label{packing-single-flow-balance}\\
&\littlesum_{\tau=1}^t (c_{\text{rt}}f^{\Sigma}_{i, \tau} + F_{i, \tau}^k\cdot c_{\text{st}}\lambda^k) \leq u_{i, t}^b\label{packing-single-battery-not-full}\\
&\littlesum_{\tau=s}^t (c_{\text{rt}}f^{\Sigma}_{i, \tau} + F_{i, \tau}^k\cdot c_{\text{st}}\lambda^k) \leq u_{i, t, s}^B, \quad2\leq s\leq t\label{packing-single-battery-full}\\
&\lambda^k \geq 0, \> f_{ij, t}\geq 0, \forall (i, j)\in E\label{packing-single-nonnegativity} 
\end{flalign}
Notice that in this formulation all the variables are on the left-hand side of the constraints, whereas all the right-hand sides are constant.

\subsection{Relation to Multi-commodity Flow}

Let $T=2$, and observe the constraints in (\ref{packing-flow-balance})--(\ref{packing-nonnegativity}). We claim that verifying whether any set of sensing rates $\lambda_{i, t}$ is feasible is at least as hard as solving a 2-commodity feasible flow problem with capacitated nodes and a single sink. To prove the claim, we first rewrite the constraints in (\ref{packing-flow-balance})--(\ref{packing-nonnegativity}) as:
\begin{align*}
&\littlesum_{(j, i)\in E}f_{ji, t} + \lambda_{i, t} = \littlesum_{(i, j)\in E}f_{ij, t}, \hspace{48pt} t\in\{1, 2\}\\
&c_{\text{rt}}\littlesum_{(j, i)\in E}f_{ji, 1} \leq b_{i, 1}+e_{i, 1}-c_{\text{st}}\lambda_{i, 1}\\
&c_{\text{rt}}\littlesum_{\tau=1}^2 \littlesum_{(j, i)\in E}f_{ji, \tau} \leq b_{i, 1}+ \littlesum_{\tau=1}^2\left(e_{i, \tau}-c_{\text{st}}\lambda_{i, \tau}\right)\\
&c_{\text{rt}}\littlesum_{(j, i)\in E}f_{ji, 2} \leq B+ e_{i, 2}-c_{\text{st}}\lambda_{i, 2}\\
&\lambda_{i, t} \geq 0, f_{ij, t}\geq 0,\; \forall i\in \{1, ..., n\}, (i, j)\in E, t\in\{1, 2\}
\end{align*}

Suppose that we are given any 2-commodity flow problem with capacitated nodes, and let:
\begin{itemize}[itemsep=-3pt, topsep=3pt, leftmargin=10pt]
\item $\lambda_{i, t}$ denote the supply of commodity $t$ at node $i$;
\item $u_{i, t}$ denote the per-commodity capacity constraint at node $i$ for commodity $t$;
\item $u_i$ denote the bundle capacity constraint at node $i$.
\end{itemize}

Choose values of $c_{\text{s}}, c_{\text{rt}}, B, b_{i, 1}, b_{i, 2}, e_{i, 1}, e_{i, 2}$ so that the following equalities are satisfied:
\begin{flalign*}
u_{i, 1 }&= \frac{1}{c_{\text{rt}}}\cdot\left( b_{i, 1}+e_{i, 1}-c_{\text{st}}\lambda_{i, 1} \right)\\
u_{i, 2} &= \frac{1}{c_{\text{rt}}}\cdot\left( B+ e_{i, 2}-c_{\text{st}}\lambda_{i, 2} \right)\\
u_i &= \frac{1}{c_{\text{rt}}}( b_{i, 1}+\littlesum_{\tau=1}^2\left(e_{i, \tau}-c_{\text{st}}\lambda_{i, \tau})\right)
\end{flalign*}
Then feasibility of the given 2-commodity flow problem is equivalent to the feasibility of (\ref{packing-flow-balance})--(\ref{packing-nonnegativity}). Therefore, any 2-commodity feasible flow problem can be stated as an equivalent problem of verifying feasibility of sensing rates $\lambda_{i, t}$ in an energy harvesting network for $T=2$.

For $T>2$, (\ref{packing-battery-not-full}) and (\ref{packing-battery-full}) are general packing constraints. If a flow graph $G_t$ in time slot $t$ is observed as a flow of a commodity indexed by $t$, then for each node $i$ the constraints (\ref{packing-battery-not-full}) and (\ref{packing-battery-full}) define capacity constraints for every sequence of consecutive commodities $s, s+1,..., t$, $1\leq s\leq t\leq T$.

Therefore, to our current knowledge, it is unlikely that the general rate assignment problem can be solved exactly in polynomial time without the use of linear programming, as there have not been any combinatorial algorithms that solve feasible 2-commodity flow exactly.

\subsection{Fractional Packing Approach}

The fractional packing problem is defined as follows \cite{Plotkin1995}:\\
\textsc{Packing: }\emph{Given a convex set $P$ for which $Ax\geq 0$ $\forall x\in P$, is there a vector $x$ such that $Ax\leq b$?} Here, $A$ is a $p\times q$ matrix, and $x$ is a $q$-length vector.

A given vector $x$ is an $\epsilon$-approximate solution to the \textsc{Packing} problem if $x\in P$ and $Ax\leq(1+\epsilon)b$. Alternatively, scaling all the constraints by $\frac{1}{1+\epsilon}$, we obtain a solution $x'=\frac{1}{1+\epsilon}x\in(\frac{1}{1+\epsilon}x_{\text{OPT}}, x_{\text{OPT}}]\subset ((1-\epsilon)x_{\text{OPT}}, x_{\text{OPT}}]$, for $\epsilon<1$, where $x_{\text{OPT}}$ is an optimal solution to the packing problem. The algorithm in \cite{Plotkin1995} either provides an $\epsilon$-approximate solution to the \textsc{Packing} problem, or it proves that no such solution exists. It\rq{}s running time depends on:
\begin{itemize}[itemsep=-3pt, topsep=3pt, leftmargin=10pt]
\item The running time required to solve $\min \{cx: x\in P\}$, where $c=y^TA$, $y$ is a given $p$-length vector, and $(.)^T$ denotes the transpose of a vector.
\item The width of $P$ relative to $Ax\leq b$, which is defined by $\rho = \max_i \max_{x\in P}\frac{a_ix}{b_i}$, where $a_i$ is the $i^{\text{th}}$ row of $A$, and $b_i$ is the $i^{\text{th}}$ element of $b$.
\end{itemize}

For a given error parameter $\epsilon>0$, a feasible solution to the problem $\min\{\beta: Ax\leq \beta b, x\in P\}$, its dual solution $y$, and $C_{\mathcal{P}}(y)=\min\{cx: c=y^TA, x\in P\}$, \cite{Plotkin1995} defines the following relaxed optimality conditions:
\begin{flalign*}
&(1-\epsilon)\beta y^Tb\leq y^TAx &(\mathcal{P}1)\\
& y^TAx - C_{\mathcal{P}}(y) \leq \epsilon(y^TAx+\beta y^Tb) &(\mathcal{P}2)
\end{flalign*}

The packing algorithm \cite{Plotkin1995} is implemented through subsequent calls to the procedure \textsc{Improve-Packing}:
\begin{algorithm}
\caption{\textsc{Improve-Packing}($x, \epsilon$)\cite{Plotkin1995}}
\begin{algorithmic}[1]
\State Initialize $\beta_0 = \max_i a_ix/b_i; \, \alpha = 4\beta_0^{-1}\epsilon^{-1}\ln(2p\epsilon^{-1});\, \sigma=\epsilon/(4\alpha \rho)$.
\While{$\max_i a_ix/b_i\geq \beta_0/2$ and $x, y$ do not satisfy ($\mathcal{P}2$)}
	\State  For each $i=1,2,..., p$: set $y_i=(1/b_i)e^{\alpha a_ix/b_i}$.
	\State  Find a min-cost point $\tilde{x}\in P$ for costs $c=y^TA$.
	\State  Update $x=(1-\sigma)x+\sigma\tilde{x}$.
\EndWhile
\State \textbf{return} $x$.
\end{algorithmic}
\end{algorithm}

The running time of the $\epsilon$-approximation algorithm provided in \cite{Plotkin1995}, for $\epsilon \in (0, 1]$, equals $O(\epsilon^{-2}\rho\log(m\epsilon^{-1}))$ multiplied by the time needed to solve $\min \{cx: c=y^TA, x\in P\}$ and compute $Ax$ (Theorem 2.5 in \cite{Plotkin1995}).

\subsubsection{Maximizing the Rates as Fractional Packing}

We demonstrated at the beginning of this section that for the $k^{\text{th}}$ iteration \textsc{Maximize-the-Rates} can be stated as (\ref{packing-single})-(\ref{packing-single-nonnegativity}). Observe the constraints (\ref{packing-single-battery-not-full}) and (\ref{packing-single-battery-full}). Since $\lambda^k$, $f_{ij, t}$ and all the right-hand sides in (\ref{packing-single-battery-not-full}) and (\ref{packing-single-battery-full}) are nonnegative, (\ref{packing-single-battery-not-full}) and (\ref{packing-single-battery-full}) imply the following inequalities:
\begin{align*}
&\forall i\in \{1, ..., n\}, t \in \{1,..., T\}:\\
& F_{i, \theta}^k\cdot c_{\text{st}}\lambda^k \leq u_{i,t}^b,
\hspace{130pt} 1\leq\theta\leq t\\
&F_{i, \theta}^k\cdot c_{\text{st}}\lambda^k \leq u_{i, t, s}^B, 
\hspace{80pt} 2\leq s\leq t,\; s\leq \theta\leq t\\
&c_{\text{rt}}\littlesum_{(j, i)\in E}f_{ji, \theta} \leq u_{i, t}^b -c_{\text{st}}\littlesum_{\tau=1}^t F_{i, \tau}^k\lambda^k,
\hspace{50pt} 1\leq \theta\leq t\\
&c_{\text{rt}}\littlesum_{(j, i)\in E}f_{ji, \theta}\leq u_{i, t, s}^B-c_{\text{st}}\littlesum_{\tau=s}^t F_{i, \tau}^k\lambda^k,
 \;2\leq s\leq t,\; s\leq \theta \leq t
\end{align*}

Therefore, we can yield an upper bound $\lambda_{\max}^k$ for $\lambda^k$:
\begin{align}
&\lambda^k\leq \lambda_{\max}^k\equiv \notag\\
& \frac{1}{c_{\text{st}}}\min_{i, t, s\geq 2}\{u_{i, t}^b:\littlesum_{\tau=1}^t F_{i, \tau}^k>0, u_{i,t,s}^B: \littlesum_{\tau=s}^t F_{i, \tau}^k>0 \}\label{packing-lambda-max}
 \end{align}
 For a fixed $\lambda^k$, the flow entering a node $i$ at time slot $t$ can be bounded as:
 \begin{align}
 &\littlesum_{(j, i)\in E}f_{ji, t}\leq u_{i, t}\equiv\notag\\
 & \frac{1}{c_{\text{rt}}}\min_{\substack{i, t_1\geq t\\ s\geq 2}}\{ u_{i, t_1}^b- c_{\text{st}}\littlesum_{\tau=1}^{t_1}F_{i, \tau}^k\lambda^k,
 u_{i, t, s}^B- c_{\text{st}}\littlesum_{\tau=s}^{t_1}F_{i, \tau}^k\lambda^k \}\label{packing-node-capacity}
\end{align}

We choose to keep only the flows $f_{ij, t}$ as variables in the \textsc{Packing} problem. Given a $\lambda^k\in [0, \lambda_{\max}^k]$, we define the convex set $P$\footnote{$P$ is determined by linear equalities and inequalities, which implies that it is convex.} via the following set of constrains:
\begin{align}
&\forall i\in \{1, ..., n\}, t \in \{1,..., T\}:\notag\\
&-\littlesum_{(j, i)\in E}f_{ji, t}  + \littlesum_{(i, j)\in E}f_{ij, t} = \lambda_{i, t}^{k-1}+ F_{i, t}^k\cdot \lambda^k \label{P-flow-balance}\\
&\littlesum_{(j, i)\in E}f_{ji, t}\leq u_{i, t}\label{P-capacity}\\
& f_{ij, t}\geq 0, \hspace{14em}\forall (i, j)\in E \label{P-nonnegativity}
\end{align}
\begin{proposition}\label{prop:P-min-cost-flow}
For $P$ described by $(\ref{P-flow-balance})-(\ref{P-nonnegativity})$ and a given vector $y$, problem $\min \{cf: c=y^TAf,\, f\in P\}$ reduces to $T$ min-cost flow problems.
\end{proposition}
\iffullresults
\begin{proof}
Constraint (\ref{P-flow-balance}) is a standard flow balance constraint at a node $i$ in a time slot $t$, whereas constraint (\ref{P-capacity}) corresponds to a node capacity constraint at the time $t$, given by (\ref{packing-node-capacity}). As there is no interdependence of flows over time slots, we get that the problem can be decomposed into subproblems corresponding to individual time slots. Therefore, to solve the problem $\min \{cf: c=y^TAf,\, f\in P\}$ for a given vector $y$, it suffices to solve $T$ min-cost flow problems, one for each time slot $t\in \{1,2,...,T\}$.
\end{proof}
\fi
The remaining packing constraints of the form $Ax\leq b$ are given by (\ref{packing-single-battery-not-full}) and (\ref{packing-single-battery-full}), where $x\equiv f$.
\begin{proposition}\label{prop:Ax-nonnegative}
$Ax\geq 0$ $\forall f\in P$.
\end{proposition}
\iffullresults
\begin{proof}
  As $f_{ij, t}\geq 0$ $\forall(i, j)\in E, t\in\{1,..., T\}$, and all the coefficients multiplying $f_{ij, t}$'s in (\ref{packing-single-battery-not-full}) and (\ref{packing-single-battery-full}) are nonnegative, the result follows immediately.
  \end{proof}
  \fi
\begin{lemma}\label{lemma:improve-packing-time}
One iteration of \textsc{Improve-Packing} for \textsc{P-Fractional} can be implemented in time
\begin{equation*}
O\left(nT^2+T\cdot MCF(n, m)\right),
\end{equation*}
where $MCF(n, m)$ denotes the running time of a min-cost flow algorithm on a graph with $n$ nodes and $m$ edges.
\end{lemma}
\iffullresults
\begin{proof}
  Since the flows over edges appear in the packing constraints only as the sum-terms of the total incoming flow of a node $i$ in a time slot $t$, we can use the total incoming flow $f_{i, t}^{\Sigma}=\littlesum_{(j, i)\in E}f_{ji, t}$ for each $(i, t)$ as variables. Reordering the terms, the packing constraints can be stated as:
\begin{flalign}
&\littlesum_{\tau=1}^t  f_{i, \tau}^{\Sigma} \leq \frac{1}{c_{\text{rt}}}(u^b_{i, t}- c_{\text{st}}\littlesum_{\tau=1}^t F_{i, \tau}^k\lambda^k), \hspace{30pt}1\leq t \leq T\label{packing-1}\\
&\littlesum_{\tau=s}^t f_{i, \tau}^{\Sigma} \leq\frac{1}{c_{\text{rt}}}( u^B_{i, t, s}- c_{\text{st}}\littlesum_{\tau=s}^t F_{i, \tau}^k\lambda^k),\; 2\leq s\leq t,\; 2\leq t\leq T\label{packing-2}
\end{flalign}
With this formulation on hand, the matrix $A$ of the packing constraints $Af^{\Sigma}\leq b$ is a $0-1$ matrix that can be decomposed into blocks of triangular matrices. To see this, first notice that for each node $i$ constraints given by (\ref{packing-1}) correspond to a lower-triangular 0-1 matrix of size $T$. Each sequence of constraints of type (\ref{packing-2}) for fixed $i$ and fixed $s\in\{2,...,T\}$, and $t \in\{s, s+1,..., T\}$ corresponds to a lower-triangular 0-1 matrix of size $T-s+1$. This special structure of the packing constraints matrix allows an efficient computation of the dual vector $y$ and the corresponding cost vector $c$. Moreover, as constraints (\ref{packing-1}, \ref{packing-2}) can be decomposed into independent blocks of constraints of the type $A_if_i^{\Sigma}\leq b_i$ for nodes $i\in\{1,...,n\}$, the dual vector $y$ and the corresponding cost vector $c$ can be decomposed into vectors $y_i, c_i$  for $i\in\{1,...,n\}$. Cost $c_{i, t}$ can be interpreted as the cost of sending 1 unit of flow through node $i$ in time slot $t$.

Observe the block of constraints $A_if_i^{\Sigma}\leq b_i$ corresponding to the node $i$. The structure of $A_i$ is as follows:
\[ \quad\quad T\left\{\begin{matrix}
1 & 0 & 0 & \cdots &0&0\\
1 & 1 & 0 & \cdots &0&0\\
\vdots & \vdots & \vdots & \ddots & \vdots&\vdots\\
1 & 1 & 1 &\cdots &1&1
\end{matrix}\right.\]
\[ T-1 \left\{\begin{matrix}
0 & 1 & 0 & \cdots &0&0\\
0 & 1 & 1 & \cdots & 0&0\\
\vdots & \vdots & \vdots & \ddots & \vdots&\vdots\\
0 & 1 & 1 & \cdots &1&1
\end{matrix}
\right. 
\]
\[\quad\quad\quad\vdots\]
\[\quad\quad 2 \left\{\begin{matrix}
0 & 0 &  0 & \cdots &1 & 0\\
0 & 0 & 0 & \cdots &1 & 1
\end{matrix}
\right. 
\]
\[\quad\quad  1\left\{\begin{matrix}
0 & 0 &  0 & \cdots &0 & 1
\end{matrix}
\right.
\] 
  As $A_i$ can be decomposed into blocks of triangular matrices, each $y_{i, j}$ in the \textsc{Improve-Packing} procedure can be computed in constant time, yielding $O\left(\frac{T(T-1)}{2}\right)=O\left(T^2\right)$ time for computing $y_i$. This special structure of $A_i$ also allows a fast computation of the cost vector $c_i$. Observe that each $c_{i, t}, t \in \{1,...,T\}$ can be computed by summing $O(T)$ terms. For example, $c_{i, 1}=\littlesum_{j=1}^T y_{i, j}$, $c_{i, 2} = c_{i, 1}-y_{i, 1}+\littlesum_{j=T+1}^{2T-1}y_{i, j}$, $c_{i, 3} = c_{i, 2}-y_{i, 2}-y_{i, T+1}+\littlesum_{j=2T}^{3T-2}y_{i, j}$, etc. Therefore, computing the costs for node $i$ takes $O(T^2)$ time. This further implies that one iteration of \textsc{Improve-Packing} takes $O\left(nT^2+T\cdot MCF(n, m)\right)$ time, where $MCF(n, m)$ denotes the running time of a min-cost flow algorithm on a graph with $n$ nodes and $m$ edges.
  \end{proof}
\fi
\begin{lemma}\label{lemma:width}
Width $\rho$ of $P$ relative to the packing constraints $(\ref{packing-single-battery-not-full})$ and $(\ref{packing-single-battery-full})$ is $O(T)$.
\end{lemma}
\iffullresults
\begin{proof}
As $u_{i, t}$ is determined by the tightest constraint in which $\littlesum_{(j, i)\in E}f_{ji, t}\equiv f^{\Sigma}_{i, t}$ appears, we have that in every constraint given by (\ref{packing-1}, \ref{packing-2}):

\begin{flalign*}
& f_{i, \theta}^{\Sigma} \leq  \frac{1}{c_{\text{rt}}}(u^b_{i, t}- c_{\text{st}}\littlesum_{\tau=1}^t F_{i, \tau}^k\lambda^k),\hspace{50pt} 1\leq \theta\leq t\\
&f_{i, \theta}^{\Sigma}\leq\frac{1}{c_{\text{rt}}}( u^B_{i, t, s}- c_{\text{st}}\littlesum_{\tau=s}^t F_{i, \tau}^k\lambda^k),\;
2\leq s\leq t, s\leq \theta\leq t
\end{flalign*}
As the sum of $f_{ij, \theta}^{\Sigma}$ over $\theta$ in any constraint from (\ref{packing-1}, \ref{packing-2}) can include at most $T$ terms, it follows that $\rho\leq\frac{T\cdot b_i}{b_i}=T$.
\end{proof}
\fi
\begin{lemma}\label{fractional-packing-max-time}
\textsc{Maximizing-the-Rates} that uses packing algorithm from \emph{\cite{Plotkin1995}} can be implemented in time: 
$\tilde{O}({T^2}\epsilon^{-2}\cdot(nT+MCF(n, m)))$, 
where $\tilde{O}$-notation ignores poly-log terms.
\end{lemma}
\iffullresults
\begin{proof}
We have from (\ref{packing-lambda-max}) that $\lambda^k\in [0, \lambda_{\max}^k]$, therefore, we can perform a binary search to find the maximum $\lambda^k$ for which both $\min \{y^TAf| f\in P\}$ is feasible and \textsc{Packing} outputs an $\epsilon$-approximate solution. Multiplying the running time of the binary search by the running time of the packing algorithm \cite{Plotkin1995}, the total running time becomes:
\begin{align*}
&O\left(\log\left( \frac{\lambda_{\max}^k}{\delta}\right)\epsilon^{-2}\rho\log(m\epsilon^{-1})\left(nT^2 + T\cdot MCF(n, m)\right) \right)\notag\\
&=
 \tilde{O}\left( \frac{T^2}{\epsilon^2}\cdot \left(nT+MCF(n, m)\right) \right).
\end{align*}
\end{proof}
\fi

\subsubsection{Fixing the Rates}

As \textsc{Maximizing-the-Rates} described in previous subsection outputs an $\epsilon$-approximate solution in each iteration, the objective of the algorithm is not to output a max-min fair solution anymore, but an $\epsilon$-approximation. 
We consider the following notion of approximation, as in \cite{fairness-in-routing}:

\begin{definition}For a problem of lexicographic maximization, say that a feasible solution given as a vector $v$ is an element-wise $\epsilon$-approximate solution, if for vectors $v$ and $v_{\text{OPT}}$ sorted in nondecreasing order $v\geq(1-\epsilon)v_{\text{OPT}}$ component-wise, where $v_{\text{OPT}}$ is an optimal solution to the given lexicographic maximization problem.
\end{definition}

Let $\Delta$ be the smallest real number that can be represented in a computer, and consider the algorithm that implements \textsc{Fixing-the-Rates} as stated below.
 \begin{algorithm}
\caption{\textsc{Fixing-the-Rates}}
\begin{algorithmic}[1]
\State Solve the following linear program:
\State$\textbf{max } \littlesum_{i=1}^n F^k_{i, t}\lambda^k_{i, t}$
\State$\hspace{1em}\textbf{s.t. } \forall i\in \{1, ..., n\}, t \in \{1,..., T\}:$
\State$\hspace{3em}\lambda_{i, t}^{k} \geq \lambda_{i, t}^{k-1} + F_{i, t}^k\cdot \lambda^k $
\State$\hspace{3em}\lambda_{i, t}^{k} \leq \lambda_{i, t}^{k-1} + F_{i, t}^k\cdot \left(\epsilon\lambda_{i, t}^{k-1}+(1+\epsilon)\lambda^k +\Delta \right)$
\State $\hspace{3em}f^{\Sigma}_{i, t}+\lambda_{i, t}^k=\sum_{(i, j)\in E}f_{ij, t}$
\State$\hspace{3em}b_{i, t+1}=\min\Big\{ B,\; b_{i, t}+e_{i, t}
- \Big(c_{\text{rt}}f_{i, t}^{\Sigma} + c_{\text{st}}\lambda^k_{i, t} \Big)\Big\}$
\State$\hspace{3em}b_{i, t}\geq0,\quad \lambda^k_{i, t} \geq 0, \quad f_{ij, t}\geq 0$
\State Let $F^{k+1}_{i, t}=F^k_{i, t}$, $\forall i, t$.
\State If $\lambda^k_{i, t}<(1+\epsilon)(\lambda_{i, t}^{k-1} + F_{i, t}^k\cdot \lambda^k)+\Delta$, set $F^{k+1}_{i, t}=0$.
\end{algorithmic}
\end{algorithm}

Assume that \textsc{Fixing-the-Rates} does not change any of the rates, but only determines what rates should be fixed in the next iteration, i.e., it only makes (global) changes to $F^{k+1}_{i, t}$. Then:
\begin{lemma}\label{lemma:fixing-the-rates-packing}
If the Steps \ref{maximizing} and \ref{fixing} in the \textsc{Water-filling-Framework} are implemented as \textsc{Maximizing-the-Rates} and \textsc{Fixing-the-Rates} from this section, then the solution output by the algorithm is an element-wise $\epsilon$-approximate solution to the lexicographic maximization of $\lambda_{i, t}\in\mathcal{R}$.
\end{lemma}
\iffullresults
\begin{proof}
The proof is by induction.

\textbf{The base case.} Observe the first iteration of the algorithm. After rate maximization, $\forall i, t: \lambda_{i, t}=\lambda^{1}\geq \dfrac{1}{1+\epsilon}\lambda^{1}_{\text{OPT}}$ and $F^1_{i, t}=1$.

Observe that in the output of the linear program of \textsc{Fixing-the-Rates}, all the rates must belong to the interval $[\lambda^1, (1+\epsilon)\lambda^1+\Delta]$. Choose any $(i, t)$ with $\lambda_{i, t}^1<(1+\epsilon)(\lambda_{i, t}^{k-1} + F_{i, t}^1\cdot \lambda^1)+\Delta=(1+\epsilon)\lambda^1+\Delta$. There must be at least one such rate, otherwise the rate maximization did not return an $\epsilon$-approximate solution. As $\littlesum_{i=1}^n F^1_{i, t}\lambda^1_{i, t}= \littlesum_{i=1}^n \lambda^1_{i, t}$ is maximum, if $\lambda_{i, t}^1$ is increased, then at least one other rate needs to be decreased to maintain feasibility. To get a lexicographically greater solution
$
\lambda_{i, t}^1
$
can only be increased by lowering the rates with the value greater than $\lambda_{i, t}^1$. Denote by $S^1_{i, t}$ the set of all the rates $\lambda^1_{j, \tau}$ such that $\lambda^1_{j, \tau}>\lambda^1_{i, t}$. In the lexicographically maximum solution, the highest value to which $\lambda^1_{i, t}$ can be increased is at most 
$\frac{1}{|S^1_{i, t}|}\left(\lambda^1_{i, t} +\littlesum_{\lambda_{j, \tau}\in S^1_{i, t}}\lambda^1_{j, \tau} \right)<(1+\epsilon)\lambda^1+\Delta,
$ which implies $\lambda_{i, t, \max}\leq (1+\epsilon)\lambda^1$. Therefore, if $\lambda_{i, t}$ is fixed to the value of $\lambda^1$, it is guaranteed to be in the $\epsilon$-range of its optimal value.

Now consider all the $(i, t)$\rq{}s with $\lambda^1_{i, t}=(1+\epsilon)\lambda^1+\Delta$. As all the rates that get fixed are fixed to a value $\lambda_{i, t}=\lambda^1\leq\lambda^1_{i, t}$, it follows that in the next iteration all the rates that did not get fixed can be increased by at least $\epsilon \lambda^1+\Delta$, which \textsc{Fixing-the-Rates} properly determines.

\textbf{The inductive step.} Suppose that up to iteration $k\geq 2$ all the rates that get fixed are in the $\epsilon$-optimal range, and observe the iteration $k$. All the rates that got fixed prior to iteration $k$ satisfy:
\begin{align*}
\lambda_{i, t}^{k} &\geq \lambda_{i, t}^{k-1} + F_{i, t}^k\cdot \lambda^k=\lambda^{k-1}_{i, t} \text{ , and}\\
\lambda_{i, t}^{k} &\leq \lambda_{i, t}^{k-1} + F_{i, t}^k\cdot \left(\epsilon\lambda_{i, t}^{k-1}+(1+\epsilon)\lambda^k +\Delta \right)=\lambda^{k-1}_{i, t}
\end{align*}
and, therefore, they remain fixed for the next iteration, as $\lambda^{k}_{i, t}=\lambda^{k-1}_{i, t}<(1+\epsilon)\lambda^{k-1}_{i, t}$.

 Now consider all the $(i, t)$\rq{}s with $F^k_{i, t}=1$. We have that:
 \begin{align*}
 \lambda^k_{i, t}&\geq \lambda^{k-1}+1\cdot \lambda^k=\littlesum_{l=1}^k \lambda^l\\
  \lambda^k_{i, t}&\leq (1+\epsilon)\left( \lambda^{k-1}+1\cdot \lambda^k \right)+\Delta=(1+\epsilon)\littlesum_{l=1}^k \lambda^l+\Delta
 \end{align*}
 Similarly as in the base case, if $\lambda^k_{i, t}<(1+\epsilon)\littlesum_{l=1}^k \lambda^l+\Delta$, let $S^k_{i, t}=\{\lambda^k_{j, \tau}:\lambda^k_{j, \tau}>\lambda^k_{i, t}\}$. There must be at least one such $(i, t)$, otherwise the rate maximization did not output an $\epsilon$-approximate solution. In any lexicographically greater solution:
 \begin{align*}\lambda^k_{i, t, \max}\leq& \dfrac{1}{|S^k_{i, t}|}\left( \lambda^k_{i, t} +\littlesum_{\lambda^k_{j, \tau}\in S^k_{i, t}}\lambda_{j, \tau} \right)\\
 <&(1+\epsilon)\littlesum_{l=1}^k \lambda^l+\Delta,
 \end{align*}
 which implies $\lambda^k_{i, t, \max}\leq (1+\epsilon)\littlesum_{l=1}^k \lambda^l$. Therefore, if  we fix $\lambda_{i, t}$ to the value $\littlesum_{l=1}^k \lambda^l$, it is guaranteed to be at least as high as $(1-\epsilon)$ times the value it gets in the lexicographically maximum solution.

 Finally, all the $(i, t)$'s with $\lambda^k_{i, t}=(1+\epsilon)\littlesum_{l=1}^k \lambda^l+\Delta$ can simultaneously increase their rates by at least $\epsilon \littlesum_{l=1}^k \lambda^l+\Delta$ in the next iteration, so it should be $F^{k+1}_{i, t}=1$, which agrees with \textsc{Fixing-the-Rates}.
\end{proof}
\fi

\begin{lemma}\label{lemma:packing-total-time}
An FPTAS for \textsc{P-Fractional} can be implemented in time:
\begin{equation*} \tilde{O}(nT(T^2\epsilon^{-2}\cdot (nT+MCF(n, m) + LP(mT, nT))),
\end{equation*}
where $LP(mT, nT)$ denotes the running time of a linear program with $mT$ variables and $nT$ constraints, and $MCF(n, m)$ denotes the running time of a min-cost flow algorithm run on a graph with $n$ nodes and $m$ edges.
\end{lemma}
\iffullresults
\begin{proof}
It was demonstrated in the proof of Lemma \ref{lemma:fixing-the-rates-packing} that in every iteration at least one rate gets fixed. Therefore, there can be at most $O(nT)$ iterations. From Lemma \ref{fractional-packing-max-time}, \textsc{Maximizing-the-Rates} can be implemented in time $\tilde{O}({T^2}{\epsilon^{-2}}\cdot (nT+MCF(n, m)))$. The time required for running \textsc{Fixing-the-Rates} is $LP(mT, nT)$, where $LP(mT, nT)$ denotes the running time of a linear program with $mT$ variables and $nT$ constraints. 
\end{proof}
\fi
\begin{note}
A linear programming framework as in \cite{Radunovic2007, sensnet--lexicographic, OSU--lexicographic} when applied to \textsc{P-Fractional} would yield a running time equal to $O(n^2T^2\cdot LP(mT, nT))$. As the running time of an iteration in our approach is dominated by $LP(mT, nT)$, the improvement in running time is at least $O(nT)$-fold, at the expense of providing an $\epsilon$-approximation.
\end{note}

%% file: fixed-fractional.tex
Suppose that we want to solve lexicographic maximization of the rates keeping both the routing and the rates constant over time. 
Observe that, as both the routing and the rates do not change over time, the energy consumption per time slot of each node $i$ is also constant over time and equal to $\Delta b_i = c_{\text{st}} \lambda_i + c_{\text{rt}}\littlesum_{(j, i)\in E}f_{ji}$.
\begin{proposition}\label{prop:determining-delta-bi}
Maximum constant energy consumption $\Delta b_i$ can be determined in time $O( T\log(\frac{b_{i, 1}+e_{i, 1}}{\delta}) )$ for each node $i\in V\backslash\{s\}$, for the total time of $O( nT\log(\frac{b_{i, 1}+e_{i, 1}}{\delta}) )$.
\end{proposition}
\iffullresults
\begin{proof}
Since the battery evolution can be stated as:
\begin{equation*}
b_{i, t+1}=\min\left\{B,\quad b_{i, t}+e_{i, t}-\Delta b_i \right\},
\end{equation*}
maximum $\Delta b_i$ for which $b_{i, t+1}\geq 0$ $\forall t\in \{1,...,T\}$ can be determined via a binary search from the interval $[0, b_{i, 1}+e_{i, 1}]$, for each node $i$.
\end{proof}
\fi

Similarly as in previous sections, let $F^k_i=0$ if the rate $i$ is fixed at the beginning of iteration $k$, and $F^k_i=1$ if it is not. Initially: $F^1_i=1$, $\forall i$. Rate maximization can then be implemented as follows:
\begin{algorithm}
\caption{\textsc{Maximizing-the-Rates}($G, F^k, b, e, k$)}
\begin{algorithmic}[1]
\State $\lambda_{\max}^k=\frac{1}{c_{\text{st}}}\min_i\{\Delta b_i -c_{\text{st}}\lambda_i^{k-1}: F^k_i = 1\}$
\Repeat{ for $\lambda^k\in [0, \lambda_{\max}^k]$, via binary search}
	\State Set the supply of node $i$ to $d_i=\lambda^{k-1}+F_i^k\lambda^k$, capacity of node $i$ to $u_i=\frac{1}{c_{\text{rt}}}(\Delta b_i-c_{\text{st}}\lambda^k)$, for each $i$
	\State Set the demand of the sink to $\littlesum_i d_i$
	\State Solve feasible flow problem on $G$
\Until{$\lambda^k$ takes maximum value for which the flow problem is feasible on $G$}
\end{algorithmic}
\end{algorithm}

The remaining part of the algorithm is to determine which rates should be fixed at the end of iteration $k$. We note that in each iteration $k$, the maximization of the rates produces a flow $f$ in the graph $G^k$ with the supply rates $\lambda_{i}^k$. Instead of having capacitated nodes, we can modify the input graph by a standard procedure of splitting each node $i$ into two nodes $i'$ and $i''$, and assigning the capacity of $i$ to the edge $(i', i'')$. This allows us to obtain a residual graph $G^{r, k}$ for the given flow.
We claim the following:
\begin{lemma}\label{lemma:fixing-fixed-fractional}
The rate $\lambda_i$ of a node $i\in G$ can be further increased in the iteration $k+1$ if and only if there is a directed path from $i$ to the sink node in $G^{r, k}$.
\end{lemma}
\iffullresults
\begin{proof}
First, observe that the only capacitated edges in $G^k$ are those corresponding to the nodes that were split. The capacity of an edge $(i', i'')$ corresponds to the maximum per-slot energy the node $i$ can spend without violating the battery non-negativity constraint. If an edge $(i', i'')$ has residual capacity of $u_{(i', i'')}^r>0$, then the node $i$ can spend additional $c_{\text{rt}}u_{(i', i'')}^r$ amount of energy keeping the battery level non-negative in all the time slots. If $(i', i'')$ has no residual capacity ($u_{(i', i'')}^r=0$), then the battery level of node $i$ reaches zero in at least one time slot, and increasing the energy consumption per time slot leads to $b_{i, t}<0$ for some $t$, which is infeasible.

$(\Leftarrow)$ Suppose that the residual graph contains no directed path from the node $i$ to the sink. By the flow augmentation theorem \cite{network_flows}, the flow from the node $i$ cannot be increased even when the flows from all the remaining nodes are kept constant. As the capacities correspond to the battery levels at the nodes, sending more flow from $i$ causes at least one node's battery level to become negative.

 $(\Rightarrow)$ Suppose that there is a directed path from $i$ to the sink, and let $u_i^r>0$ denote the minimum residual capacity of the edges (split nodes) on that path. Then each node on the path can spend at least $c_{\text{rt}}u_{i}^r$ amount of energy maintaining feasibility. Let $U$ denote the set of all the nodes that have a directed path to the sink in $G^{r, k}$. Then increasing the rate of each node $i\in U$ by $\Delta \lambda = \dfrac{\min_i u_i^rc_{\text{rt}}}{c_{\text{st}}+nc_{\text{rt}}}>0$ and augmenting the flows of $i\in U$ over their augmenting paths in $G^{r, k}$ each node on any augmenting path spends at most $\min_i u_i^rc_{\text{rt}}$ amount of energy, which is at most equal to the energy the node is allowed to spend maintaining feasibility.
\end{proof}
\fi
\begin{lemma}\label{lemma:fixed-fractional-time}
\textsc{Water-filling-Framework} for \textsc{P-Fixed-Fractional} can be implemented in time
\begin{equation*}O( n\log(\frac{b_{i, 1}+e_{i, 1}}{\delta}) (T + MF(n, m))),
\end{equation*}
where $MF(n, m)$ denotes the running time of a max-flow algorithm for a graph with $n$ nodes and $m$ edges.
\end{lemma}\iffullresults
\begin{proof}
From Proposition \ref{prop:determining-delta-bi}, determining the values of $\Delta b_i$ for $i\in V\backslash \{s\}$ can be implemented in time $O(nT\log(\frac{b_{i, 1}+e_{i, 1}}{\delta}))$.

Running time of an iteration of \textsc{Water-filling-Frame\-work} is determined by the running times of \textsc{Maximizing-the-Rates} and \textsc{Fixing-the-Rates}. Each iteration of the binary search in \textsc{Maximizing-the-Rates} constructs and solves a feasible flow problem, which is dominated by the time required for running a max-flow algorithm that solves feasible flow problem on the graph $G$. Therefore, \textsc{Maximizing-the-Rates} can be implemented in time $O(\log(\frac{b_{i, 1}+e_{i, 1}}{\delta}) MF(n, m))$, where $MF(n, m)$ denotes the running time of a max-flow algorithm.

\textsc{Fixing-the-Rates} constructs a residual graph $G^{r, k}$ and runs a breadth-first search on this graph, which can be implemented in time $O(n+m)$ ($=O(MF(n, m))$ for all the existing max-flow algorithms).

Every iteration of \textsc{Water-filling-Framework} fixes at least one of the rates $\lambda_i$, $i\in V\backslash\{s\}$, which implies that there can be at most $n$ iterations.

Therefore, the total running time is
\begin{equation*}O( n\log(\frac{b_{i, 1}+e_{i, 1}}{\delta}) (T + MF(n, m))).
\end{equation*}
\end{proof}\fi

%% file: hardness.tex
In this section we demonstrate that solving \textsc{P-Unsplit\-table-Find} and \textsc{P-Tree-Find} is NP-hard for both problems. Moreover, we show that it is NP-hard to obtain an approximation ratio better than $\Omega(\log n)$ for \textsc{P-Tree-Find}. For \textsc{P-Unsplittable-Find}, we design an efficient combinatorial algorithm for a relaxed version of this problem--it determines a time-invariable unsplittable routing that maximizes the minimum rate.

\subsection{Unsplittable Routing}
\begin{lemma}\label{lemma:unsplittable--hardness}
\textsc{P-Unsplittable-Find} is NP-hard.
\end{lemma}
\begin{proof}
The proof of NP-hardness for \textsc{P-Unsplittable-Find} is a simple extension of the proof of NP-hardness for max-min fair unsplittable routing provided in\cite{fairness-in-routing}. We use the same reduction as in \cite{fairness-in-routing}, derived from the non-uniform load balancing problem \cite{lenstra1990approximation}. From \cite{lenstra1990approximation, fairness-in-routing}, the following problem is NP-hard:\\
 \textsc{P-Non-uniform-Load-Balancing}: Let $J=\{J_1,...,J_k\}$ be a set of jobs, and $M=\{M_1,...,M_n\}$ be a set of machines. Each job $J_i$ has a time requirement $r_i\in\{1/2, 1\}$, and the sum of all the job requirements is equal to $n$: $\sum_{i=1}^k r_i=n$. Each job $J_i\in J$ can be run only on a subset of the machines $S_i\subset M$. Is there an assignment of jobs to machines, such that the sum requirement of jobs assigned to each machine $M_j$ equals 1?
 
  For a given instance of \textsc{P-Non-uniform-Load-Balancing} we construct an instance of \textsc{P-Unsplittable-Find} as follows (Fig.~\ref{fig:unsplittable--hardness}). Let $T=1$, and $c_{\text{st}}=c_{\text{rt}}=1$. Create a node $J_i$ for each job $J_i\in J$, a node $M_j$ for each machine $M_j\in M$, and add an edge $(J_i, M_j)$ if $M_j\in S_i$. Connect all the nodes $M_j\in M$ to the sink. Let available energies at the nodes be $b_{J_i}=r_i$, $b_{M_j}=2$. 
   \begin{figure}[h]
\centering
\includegraphics[width=0.6\linewidth]{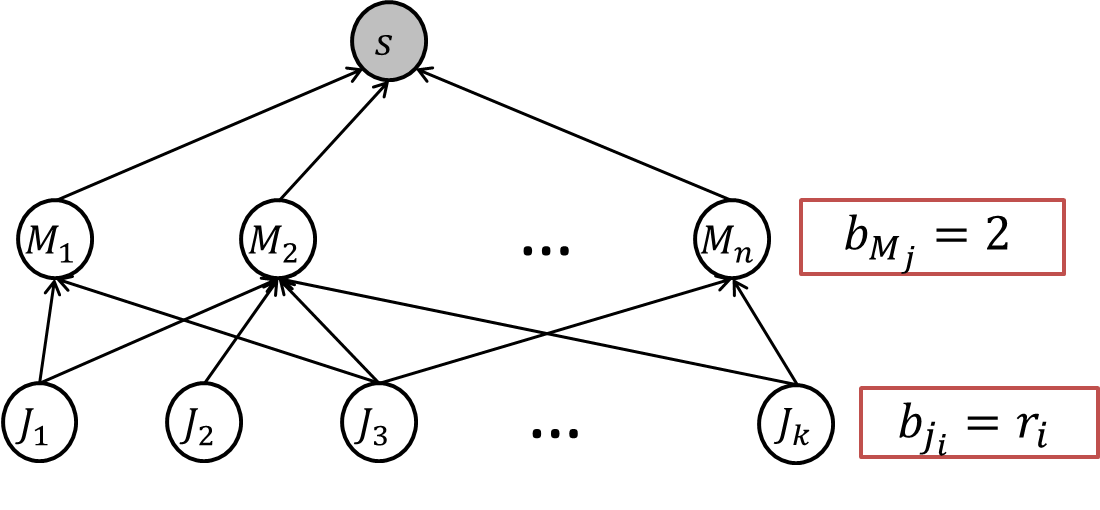}\vspace{-10pt}
\caption{A reduction from \textsc{P-Non-uniform-Load-Balancing} for proving NP-hardness of \textsc{P-Unsplittable-Find}. Jobs are represented by nodes $J_i$, machines by nodes $M_j$, and there is an edge from $J_i$ to $M_j$ if job $J_i$ can be executed on machine $M_j$. Each job $J_i$ has time requirement $r_i\in\{1/2, 1\}$, and $\sum_{i=1}^k J_i=n$. Available energies at the nodes are shown in the boxes next to the nodes. If at the optimum of \textsc{P-Unsplittable-Find} $\lambda_{J_i}=r_i$ and $\lambda_{M_j}=1$, then there is an assignment of jobs to the machines such that the sum requirement of jobs assigned to each machine equals 1.}
\label{fig:unsplittable--hardness}
\end{figure}

Suppose that the instance of \textsc{P-Non-uniform-Load-Ba\-lancing} is a "yes" instance, i.e., there is an assignment of jobs to machines such that the sum requirement of jobs assigned to each machine equals 1. Observe the following rate assignment: $\lambda^*=\{\lambda_{J_i}=r_i, \lambda_{M_j}=1\}$. This rate assignment is feasible only for the unsplittable routing in which $M_j$'s descendants are the jobs assigned to $M_j$ in the solution for \textsc{P-Non-uniform-Load-Balancing}. Moreover, as in this rate assignment all the nodes spend all their available energies and since $\sum_{i=1}^k b_{J_i}=\sum_{i=1}^k{r_i}=n$, it is not hard to see that this is the lexicographically maximum rate assignment that can be achieved for any instance of \textsc{P-Non-uniform-Load-Balancing}. If the instance of \textsc{P-Non-uniform-Load-Balancing} is a "no" instance, then \textsc{P-Unsplittable-Find} at the optimum necessarily produces a rate assignment that is lexicographically smaller than $\lambda^*$.

Therefore, if \textsc{P-Unsplittable-Find} can be solved in polynomial time, then \textsc{P-Non-uniform-Load-Balancing} can also be solved in polynomial time.
  \end{proof}
  As the proof of Lemma \ref{lemma:unsplittable--hardness} is constructed for $T=1$, it follows that \textsc{P-Unsplittable-Find} is NP-hard for general $T$, in either time-variable or time-invariable setting.

On the other hand, determining a time-invariable unsplittable routing that guarantees the maximum value of the minimum sensing rate over all time-invariable unsplittable routings is solvable in polynomial time, and we provide a combinatorial algorithm that solves it below. 

We first observe that in any time-invariable unsplittable routing, if all the nodes are assigned the same sensing rate $\lambda$, then every node $i$ spends a fixed amount of energy $\Delta b_i$ per time slot equal to the energy spent for sensing and sending own flow and for forwarding the flow coming from the descendant nodes: 
$\Delta b_i = \lambda\left( c_{\text{st}} + c_{\text{rt}}D_{i, t}  \right)$.

The next property we use follows from the integrality of the max flow problem with integral capacities (see, e.g., \cite{network_flows}). This property was stated as a theorem in \cite{Kleinberg96} for single-source unsplittable flows, and we repeat it here for the equivalent single-sink unsplittable flow problem:
\begin{theorem}\label{th-integrality-of-flow}
\emph{\cite{Kleinberg96}} Let $G=(N, E)$ be a given graph with the predetermined sink node $s$. If the supplies of all the nodes in the network are from the set $\{0, \lambda\}$, $\lambda>0$, and the capacities of all the edges/nodes are integral multiples of $\lambda$, then: if there is a fractional flow of value $f$, there is an unsplittable flow of value at least $f$. Moreover, this unsplittable flow can be found in polynomial time.
\end{theorem}
\begin{note} For the setting of Theorem \ref{th-integrality-of-flow}, any augmenting-path or push-relabel based max flow algorithm produces a flow that is unsplittable, as a consequence of the integrality of the solution produced by these algorithms. We will assume that the used max-flow algorithm has this property.
\end{note}

The last property we need is that our problem can be formulated in the setting of Theorem \ref{th-integrality-of-flow}. We observe that for a given sensing rate $\lambda$, each node spends $c_{\text{st}}\lambda$ units of energy for sensing, whereas the remaining energy can be used for routing the flow originating at other nodes. Therefore, for a given $\lambda$, we can set the supply of each node $i$ to $\lambda$, set its capacity to $u_i=({\Delta b_i-c_{\text{st}}\lambda})/{c_{\text{rt}}}$ (making sure that $\Delta b_i-c_{\text{st}}\lambda\geq0$), and observe the problem as the feasible flow problem. For any feasible \emph{unsplittable} flow solution with all the supplies equal to $\lambda$, we have that flow through every edge/node equals the sum flow of all the routing paths that contain that edge/node. As every path carries a flow of value $\lambda$, the flow through every edge/node is an integral multiple of $\lambda$. Therefore, to verify whether it is feasible to have a sensing rate of $\lambda$ at each node, it is enough to down-round all the nodes\rq{} capacities to the nearest multiple of $\lambda$: $u_i=\lambda\cdot\left\lfloor({\Delta b_i-c_{\text{st}}\lambda})/({c_{\text{rt}}\lambda})\right\rfloor$, and apply the Theorem \ref{th-integrality-of-flow}.

An easy upper bound for $\lambda$ is 
$\lambda_{\max}=\min_i{\Delta b_i}/{c_{\text{st}}}$, 
which follows from the battery nonnegativity constraint. The algorithm becomes clear now: 
 \begin{algorithm}
\caption{\textsc{Maxmin-Unsplittable-Routing}($G, b, e$)}
\begin{algorithmic}[1]
\State Perform a binary search for $\lambda\in[0, \lambda_{\max}]$.
\State For each $\lambda$ chosen by the binary search set node supplies to $\lambda$ and node capacities to $u_i=\lambda\cdot\left\lfloor{(\Delta b_i-c_{\text{st}}\lambda})/({c_{\text{rt}}\lambda})\right\rfloor$. Solve feasible flow problem.
\State Return the maximum feasible $\lambda$.
\end{algorithmic}
\end{algorithm}

\begin{lemma}\label{lemma:unsplittable-time}
The running time of \textsc{Maxmin-Unsplittable-Routing} is 
$O(\log(\min_i({b_{i, 1}+e_{i, 1}})/({c_{\text{st}}\delta}))(MF(n+1, m)))$,
where $MF(n, m)$ is the running time of a max-flow algorithm on an input graph with $n$ nodes and $m$ edges.
\end{lemma} 

\subsection{Routing Tree}

If it was possible to find the (either time variable or time-invariable) max-min fair routing tree in polynomial time for any time horizon $T$, then the same result would follow for $T=1$. It follows that if \textsc{P-Tree-Find} NP-hard for $T=1$, it is also NP-hard for any $T>1$. Therefore, we restrict our attention to $T=1$. 
\iffullresults

Assume w.l.o.g. $e_{i, 1}=0$ $\forall i\in V\backslash\{s\}$.
Let $\mathcal{T}$ denote a routing tree on the given graph $G$, and $D^{\mathcal{T}}_{i}$ denote the number of descendants of a node $i$ in the routing tree $\mathcal{T}$. Maximization of the common rate $\lambda_i=\lambda$ over all routing trees can be stated as:
 \begin{equation}
\max_{\mathcal{T}}\min_{i\in N} {b_i}/({c_{\text{st}}+ c_{\text{rt}}D^{\mathcal{T}}_{i}}) \label{MLT}
\end{equation}
This problem is equivalent to maximizing the network lifetime for $\lambda_i=1$ $\forall i\in V\backslash\{s\}$ as studied in \cite{infocom2005}. This problem, which we call \textsc{P-Maximum-Lifetime-Tree}, was proved to 
be NP-hard in \cite{infocom2005} using a reduction from the \textsc{Set-Cover} problem \cite{Kar72}. The instance used in \cite{infocom2005} for showing the NP-hardness of the problem has the property that the equivalent problem of finding a tree with the lexicographically maximum rate assignment, \textsc{P-Tree-Find}, is such that at the optimum $\lambda_1=\lambda_2=...=\lambda_n=\lambda$. Therefore, \textsc{P-Tree-Find} is also NP-hard. 
\fi

\iffullresults
We will strengthen the hardness result here and show that the lower bound on the approximation ratio for the \textsc{P-Tree-Find} problem is $\Omega(\log n)$. Notice that because we are using the instance for which at the optimum $\lambda_i=\lambda$ $\forall i$, the meaning of the approximation ratio is clear. In general, the optimal routing tree can have a rate assignment with distinct values of the rates, in which case we would need to consider an approximation to a vector $\{\lambda_{i}\}_{i\in\{1,...,n\}}$. However, we note that for any reasonable definition of approximation (e.g., element-wise or prefix-sum as in \cite{fairness-in-routing}) our result for the lower bound is still valid. As for the instance we use \textsc{P-Tree-Find} is equivalent to the \textsc{P-Maximum-Lifetime-Tree} problem, the lower bound applies to both problems.

We extend the reduction from the \textsc{Set-Cover} problem used in \cite{infocom2005} to prove the lower bound on the approximation ratio. In the \textsc{Set-Cover} problem, we are given elements $1, 2,..., n^*$ and sets $S_1, S_2, ..., S_m\subset \{1,2,..., n^*\}$. The goal is to determine the minimum number of sets from $S_1,...,S_m$ that cover all the elements $\{1,...,n^{*}\}$. Alternatively, the problem can be recast as a decision problem that determines whether there is a set cover of size $k$ or not. Then the minimum set cover can be determined by finding the smallest $k$ for which the answer is "yes". 

Suppose that there exists an approximation algorithm that solves \textsc{P-Tree-Find} (or \textsc{P-Maximum-Lifetime-Tree}) with the approximation ratio $r$. 
For a given instance of \textsc{Set-Cover}, construct an instance of \textsc{P-Tree-Find} as in Fig.~\ref{app-ratio-tree} and denote it by $G$. This reduction is similar to the reduction used in \cite{infocom2005}, with modifications being made by adding line-topology graphs, and by modifying the node capacities appropriately to limit the size of the solution to the corresponding \textsc{Set-Cover} problem. Let $l^x$ denote a directed graph with line topology of size $x$. Assume that all the nodes in any $l^x$ have capacities that are non-constraining. 
By the same observations as in the proof of NP-completeness of \textsc{P-Maximum-Lifetime-Tree} \cite{infocom2005}, if there is a routing tree that achieves $\lambda=1$, then there is a set cover of size $k$ for the given input instance of \textsc{Set-Cover}.

 \begin{figure}[h]
\centering
\includegraphics[width=0.7\linewidth]{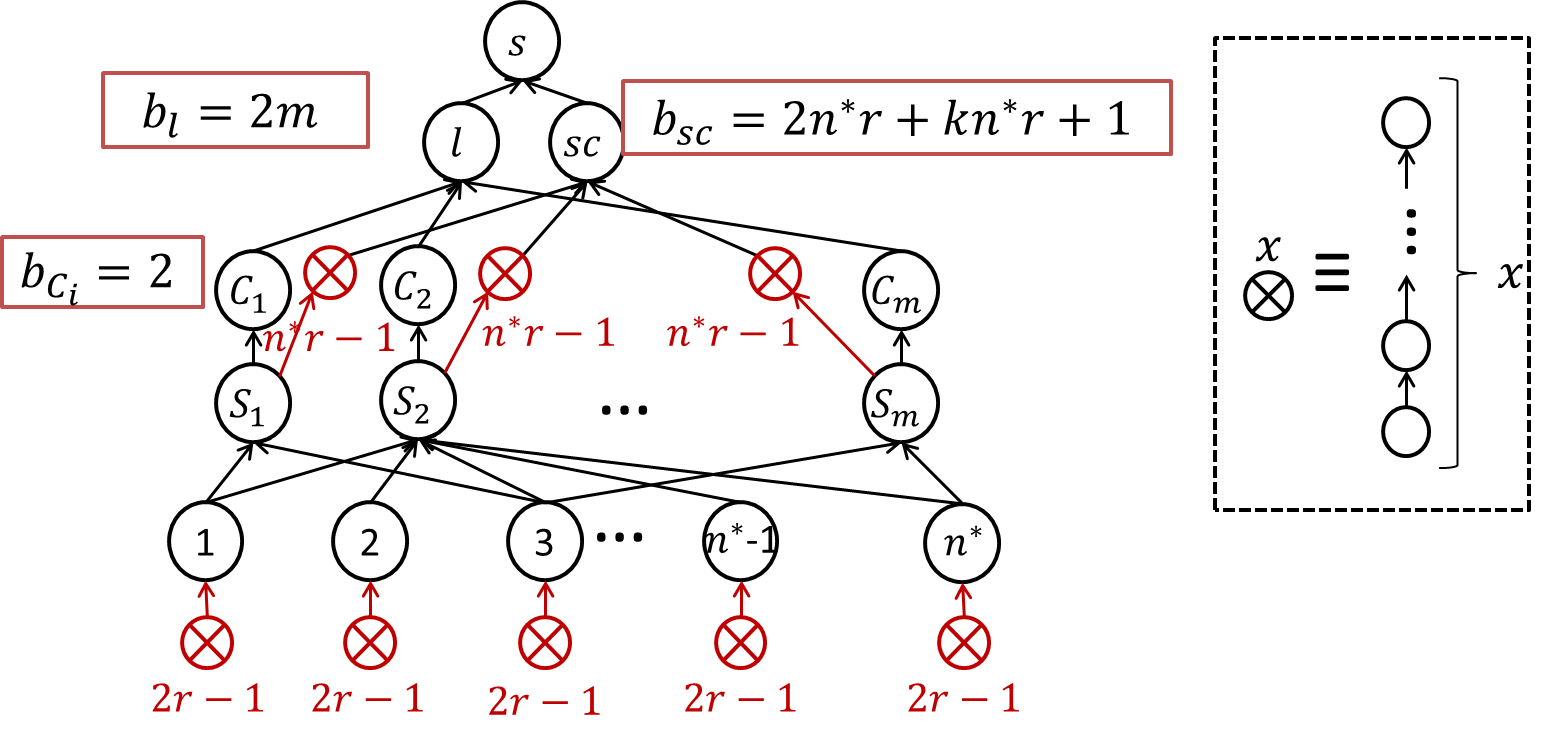}\vspace{-10pt}
\caption{A lower bound on the approximation ratio for \textsc{P-Tree-Find}. Nodes $1,2,...,n^*$ correspond to the elements, whereas nodes $S_1, S_2,..., S_m$ correspond to the sets in the \textsc{Set-Cover} problem. An element node $i$ is connected to a set node $S_j$ if in the \textsc{Set-Cover} problem $i\in S_j$. If there is a set cover of size $k$, then at $\lambda=1$ all the non-set-cover nodes are connected to the tree rooted at the node $l$, whereas all the set cover nodes and all the element nodes are in the tree rooted at $sc$. The line-topology graphs represented by crossed circles are added to limit the size of an approximate solution to the \textsc{Set-Cover} problem.}\vspace{-10pt}
\label{app-ratio-tree}
\end{figure}

Now observe a solution that an approximation algorithm with the ratio $r$ would produce, that is, an algorithm for which $\frac{1}{r}\leq\lambda\leq 1$ when $\lambda_{\text{OPT}}=1$. 

\begin{lemma}\label{leaf-nodes}
In any routing tree for which $\frac{1}{r}\leq\lambda\leq 1$, each node $C_j$ can have at most one descendant.
\end{lemma}
\begin{proof}
Suppose that there is some routing tree $T$ in which some $C_j$, $j=\{1,...,m\}$ has more than 1 descendants. Then $C_j$ must have at least one element node as its descendant. But if $C_j$ has an element node as its descendant, then the line-topology graph connected to that element node must also be in $C_j$'s descendant list, because $T$ must contain all the nodes, and a line-topology graph connected to the element node has no other neighbors. Therefore, $C_j$ has at least $2r+1$ descendants. If $\lambda\geq\frac{1}{r}$, then the energy consumption at node $C_j$ is $\frac{2r+2}{r}>2$. But the capacity of the node $C_j$ is 2, which is strictly less than the energy consumption; therefore, a contradiction.
\end{proof}

Lemma \ref{leaf-nodes} implies that if there is a routing tree that achieves $\frac{1}{r}\leq\lambda\leq 1$, then all the element nodes will be connected to the tree rooted at $sc$ through the set nodes they belong to. Therefore, the subtree rooted at $sc$ will correspond to a set cover. The next question to be asked is how large can this set cover be (as compared to $k$)? The next lemma deals with this question.

\begin{lemma}\label{sc-size}
If there is a routing tree $\mathcal{T}$ that achieves $\frac{1}{r}\leq\lambda\leq 1$, then the subtree rooted at $sc$ in $\mathcal{T}$ contains at most $p\leq3rk$ nodes.
\end{lemma}
\begin{proof}
Let $\mathcal{T}$ be a routing tree that achieves $\frac{1}{r}\leq\lambda\leq 1$.

The capacity of the node $sc$ determines the number of the set nodes that can be connected to $sc$. As all the element nodes (and line-topology graphs connected to them) are in the subtree rooted at $sc$, when there are $p$ set nodes connected to $sc$, $sc$ has $2n^*r+pn^*r$ descendants. As each node has $\frac{1}{r}\leq\lambda\leq 1$ sensing rate, the energy consumption at the node $sc$ is $e_{sc}=(2n^*r+pn^*r+1)\lambda$. For the solution to be feasible, it must be $e_{sc}\leq b_{sc}$. Therefore:
\begin{align*}
&(2n^*r+pn^*r+1)\lambda\leq2n^*r+kn^*r+1
\end{align*}
\begin{align*}
\Leftrightarrow\>\> p\leq&\frac{1}{\lambda}\cdot \frac{2n^*r+kn^*r+1}{n^*r}-\frac{2n^*r+1}{n^*r}\\
=&\frac{1}{\lambda}\left(2+k+\frac{1}{n^*r}\right)-2-\frac{1}{n^*r}
\end{align*}
As $\lambda\geq\frac{1}{r}$:
$
p\leq (2+k)r+\frac{1}{n^*}-2-\frac{1}{n^*r}\leq (2+k)r\leq k\cdot3r
$,
where the last inequality comes from $k\geq 1$.
\end{proof}

The last lemma implies that if we knew how to solve \textsc{P-Tree-Find} in polynomial time with the approximation ratio $r$, then for an instance of \textsc{Set-Cover} we could run this algorithm for $k=\{1,2,...,m-1\}$ (verifying whether $k=m$ is a set cover is trivial) and find a $3r$-approximation for the minimum set cover, which is stated in the following lemma.
\begin{lemma}\label{r-3r}
If there is a polynomial-time $r$-approximation algorithm for \textsc{P-Tree-Find}, then there is a polynomial-time $3r$-approximation algorithm for \textsc{Set-Cover}.
\end{lemma}
\begin{proof}
Suppose that there was an algorithm that solves \textsc{P-Tree-Find} in polynomial time with some approximation ratio $r$. For a given instance of \textsc{Set-Cover} construct an instance of \textsc{P-Tree-Find} as in Fig.~\ref{app-ratio-tree}. Solve (approximately) \textsc{P-Tree-Find} for $k\in\{1,..., m-1\}$. In all the solutions, it must be $\lambda\leq 1$. Let $k_m$ denote the minimum $k\in\{1,..., m-1\}$ for which $\lambda\geq\frac{1}{r}$. Then the minimum set cover size for the input instance of \textsc{Set-Cover} is $k^*\geq k_m$, otherwise there would be some other $k_m\rq{}<k_m$ for which $\lambda\geq\frac{1}{r}$. From Lemmas \ref{leaf-nodes} and \ref{sc-size}, the solution to the constructed instance of \textsc{P-Tree-Find} corresponds to a set cover of size $p\leq 3r\cdot k_m$ for the input instance. But this implies $p\leq 3r\cdot k^*$, and, therefore, the algorithm provides a $3r$-approximation to the \textsc{Set-Cover}.
\end{proof}\fi
\begin{theorem}\label{thm:tree-lower-bound}
The lower bound on the approximation ratio of \textsc{P-Tree-Find} is $\Omega(\log n)$.
\end{theorem}
\iffullresults
\begin{proof}
The lower bound on the approximation ratio of \textsc{Set-Cover} was shown to be $\Omega(\log n)$ in \cite{yannakakis}.

The proof for the lower bound on the approximation ratio given in \cite{yannakakis} was derived assuming a polynomial relation between $n^*$ and $m$. Therefore, the lower bound of $\Omega(\log n^*)$ holds for $m={n^*}^{c^*}$, where $c^* \in \mathbb{R}$ is some positive constant. Assume that $n^*\geq3$. The graph given for an instance of \textsc{Set-Cover} (as in Fig.~\ref{app-ratio-tree}) contains
$
n=2rn^*+mrn^*+3
\leq r{n^*}^{c'}
$
nodes, for some other constant $c\rq{}>1$. Therefore: $
n^*\geq\sqrt[c']{\dfrac{n}{r}}.
$
As $r\geq \dfrac{1}{3}c \log n^*$, it follows that:
\begin{gather*}
r\geq \frac{1}{3}c\log\sqrt[c\rq{}]{\frac{n}{r}}
=\frac{c}{3c'}(\log n - \log r)\\
\Leftrightarrow \>\> \frac{c}{3c'}\log r + r \geq \frac{c}{3c'} \log n
\Rightarrow r \geq c''\log n,
\end{gather*} 
for some $c''\in \mathbb{R}$.
\end{proof}\fi

%% file: conclusion.tex
\begin{table*}[t!]
\small
\caption{Running times of the algorithms for the \protect\textsc{Water-filling-Framework} implementation.}
\centering
\begin{tabular}{|c|m{0.3\linewidth}| m{0.16\linewidth}|m{0.3\linewidth}|}
\hline
&\multicolumn{1}{c|}{\multirow{2}{*}{\textsc{Maximizing-the-Rates}}} & \multicolumn{1}{c|}{\multirow{2}{*}{\textsc{Fixing-the-Rates}}}& \multicolumn{1}{c|}{\multirow{2}{*}{\textbf{Total}}}\\[10pt]
\hline
 \textsc{P-Unsplittable-Find} & $O(nT\log( \frac{B+\max_{i, t} e_{i, t}}{\delta c_{\text{st}}}))$ & $O(mT)$ & $O(nT(nT\log(\frac{B+\max_{i, t} e_{i, t}}{(\delta c_{\text{st}})} )+mT))$\\
 \textsc{P-Fixed-Fractional} & $O( n\log(\frac{b_{i, 1}+e_{i, 1}}{\delta}) (T + MF(n, m)))$ & $O(m)$ & $O( n\log(\frac{b_{i, 1}+e_{i, 1}}{\delta}) (T + MF(n, m)))$\\
 \textsc{P-Fractional} & $\tilde{O}({T^2}\epsilon^{-2}\cdot(nT+MCF(n, m)))$ & $LP(mT, nT)$ & $\tilde{O}(nT(T^2\epsilon^{-2}\cdot (nT+MCF(n, m) + LP(mT, nT)))$\\
\hline
\end{tabular}\vspace{-10pt}
\label{table:running-times}
\end{table*}

This \iffullresults paper \fi presents a comprehensive algorithmic study of the max-min fair rate assignment and routing problems in energy harvesting networks with predictable energy profile. We develop algorithms for the \textsc{Water-filling-Framework} implementation under various routing types. The running times of the developed algorithms are summarized in Table~\ref{table:running-times}. The algorithms provide important insights into the structure of the problems, and can serve as benchmarks for evaluating distributed and approximate algorithms possibly designed for unpredictable energy profiles.

The results reveal interesting trade-offs between different routing types. While we provide an efficient algorithm that solves the rate assignment problem in any time variable or time-invariable unsplittable routing or a routing tree, we also show that determining a routing with the lexicographically maximum rate assignment for any of these settings is NP-hard. On the positive side, we are able to construct a combinatorial algorithm that determines a time-invariable unsplittable routing which maximizes the minimum sensing rate assigned to any node in any time slot.

 Fractional time-variable routing provides the best rate assignment (in terms of lexicographical maximization), and both the routing and the rate assignment are determined jointly by one algorithm. However, as demonstrated in Section \ref{section:fractional}, the problem is unlikely to be solved optimally without the use of linear programming, incurring a high running time. While we provide an FPTAS for this problem, reducing the algorithm running time by a factor of $O(nT)$ (as compared to the framework of \cite{Radunovic2007, sensnet--lexicographic, OSU--lexicographic}), the proposed algorithm still requires solving $O(nT)$ linear programs.
 
If fractional routing is restricted to be time-invariable and with constant rates, the problem can be solved by a combinatorial algorithm, which we provide in Section \ref{section:fixed fractional}. However, as discussed in the introduction, constant sensing rates often result in the underutilization of the available energy.

There are several directions for future work. For example, extending the model to incorporate the energy consumption due to the control messages exchange would provide a more realistic setting. Moreover, designing algorithms for unpredictable energy profiles that can be implemented in an online and/or distributed manner is of high practical significance.